\documentclass[runningheads,a4paper]{llncs}
\usepackage[utf8x]{inputenc}
\usepackage[T1]{fontenc}
\usepackage{fixltx2e}
\usepackage{graphicx}
\usepackage{longtable,pdflscape}
\usepackage{float}
\usepackage{wrapfig}
\usepackage{soul}
\usepackage{listings}
\usepackage{textcomp}
\usepackage{multirow}
\usepackage{latexsym}
\usepackage{rotating}
\usepackage{hyperref}
\usepackage{amsmath}
\usepackage{amsfonts}
\usepackage{color,colortbl}
\usepackage{comment}
\usepackage{stmaryrd}
\usepackage{tabularx}
\usepackage{ltxtable}
\usepackage[matha,
  mathb,
  mathx
]{mathabx}
\usepackage{amssymb}
\usepackage{caption}
\usepackage{subcaption}
\usepackage{cancel}
\usepackage{paralist}

\usepackage{tikz}
\usepackage{pgfplots}
\usetikzlibrary{decorations.pathreplacing}

\usepackage{array}
\newcolumntype{C}[1]{>{\centering\arraybackslash}p{#1}}

\newcommand{\assale}[1]{{\color{blue}Assalé: #1}}
\newcommand{\ploc}[1]{{\color{green}Ploc: #1}}

\newcommand{\Var}{\mathcal{V}}

\newcommand{\Q}{\mathcal{Q}}
\newcommand{\AQ}{\mathcal{AQ}}
\newcommand{\QA}{\mathcal{QA}}
\newcommand{\ZQ}[2]{\mathcal{Z}^{#1}_{\mathcal{Q}^{#2}}}

\newcommand\rr{\mathbb{R}}

\newcommand\ux{\begin{pmatrix} x\\1 \end{pmatrix}}
\newcommand\uxt{\begin{pmatrix} x\\1 \end{pmatrix}^\intercal}
\newcommand\tr{\operatorname{tr}}
\newcommand\Mat{\mathbf{M}}
\renewcommand\st{\operatorname{s.t.}}
\newcommand\cc{\mathbf{C}}
\newcommand\ccm{\mathbf{C}^m}

\newcommand\ccmp{\mathbf{C}^{m+p}}
\newcommand\CCm{\mathcal{C}^m}
\newcommand\argmin{\operatorname{argmin}}
\newcommand\midop{\operatorname{mid}}
\newcommand\ext{\operatorname{ext}}
\newcommand\INT{\mathcal{I}}
\newcommand\sint{\sqcup_{\INT}}
\newcommand\bINT{\overline{\INT}}

\newcommand\best[1]{\cellcolor[gray]{0.8}#1}

\title{Quadratic Zonotopes}
\subtitle{An extension of Zonotopes to Quadratic Arithmetics}
\author{Assalé Adjé, Pierre-Loïc Garoche, Alexis Werey}
\hypersetup{
  pdfkeywords={},
  pdfsubject={},
  pdfcreator={Emacs Org-mode version 7.9.2}}
  
\institute{Onera, the French Aerospace Lab, France\\Université de Toulouse,
  Toulouse, France}
\begin{document}
\mainmatter
\maketitle

% \setcounter{tocdepth}{3}
% \tableofcontents
% \vspace*{1cm}

\begin{abstract}
Affine forms are a common way to represent convex sets of
$\mathbb{R}$ using a base of error terms $\epsilon \in [-1, 1]^m$. Quadratic
forms are an extension of affine forms enabling the use of quadratic error terms
$\epsilon_i \epsilon_j$. 

In static analysis, the zonotope domain, a relational abstract domain based on
affine forms has been used in a wide set of settings, e.g. set-based simulation
for hybrid systems, or floating point analysis, providing relational
abstraction of functions with a cost linear in the number of errors terms.

In this paper, we propose a quadratic version of zonotopes.
% and illustrated their
%applicability on representative examples.
%We exhibit a hierachy of abstractions between forms and affine sets.
%depending on the partial order used. 
We also present a new algorithm based on semi-definite
programming to project a quadratic zonotope, and therefore quadratic forms, to
intervals. All presented material has been implemented and applied on
representative examples.

%, galois connection}
\end{abstract}
% \assale{
% On ne devrait pas tout simplement zapper la partie affine forms/zonotopes en disant qu'une 
% vraie comparaison se fera dans un futur work...on fait juste une comparaison pratique juste sur les dessins . 
% }
%\assale{Tu dois: (verifier les preuves avec les nouvelles notations (ext), les
%  ajouter en annexe, faire des hyperliens pour les referencer}
\section{Affine arithmetics and Static Analysis}
\label{sec-1}
\paragraph{Context.}
Affine arithmetics was introduced in the 90s by Comba and
Stolfi~\cite{Comba93affinearithmetic} as an alternative to interval
arithmetics, allowing to avoid some pessimistic computation like the
cancellation:
\[
x - x = [a, b] -_\mathcal{I} [a,b] = [a - b, b - a ] \neq [0, 0]
\] 
It relies on a representation of convex subsets of $\mathbb{R}$ keeping
dependencies between variables: e.g. $x \in [-1, 1]$ will represented as $0 + 1
* \epsilon_1$ while another variable $y \in [-1, 1]$ will be represented by
another $\epsilon$ term: $y = 0 + 1 * \epsilon_2$. Therefore $x - x$ will be
precisely computed as $ \epsilon_1 - \epsilon_1 = 0$ while $x - y$ will result
in $\epsilon_1 - \epsilon_2$, i.e. denoting the interval $[-2,2]$.

In static analysis, affine forms lifted to abstract environments, as vectors of
affine forms, are a friendly alternative to costly relational domains. They
provide cheap and scalable relational abstractions: their complexity is linear
in the number of error terms -- the $\epsilon_i$ -- while most relational
abstract domains have a complexity at least cubic. Since their geometric
concretization characterizes a zonotope, i.e. a symmetric convex
polytope, they are commonly known as zonotopic abstract domains.

However since zonotopes are not fitted with lattice structure, their use in pure
abstract interpretation using a Kleene iteration schema is not common. The
definition of an abstract domain based on affine forms requires the definition
of an upper bound and lower bound operators since no least upper bound and
greatest lower bound exist in general. Choices vary from the computation of a
precise minimal upper bound to a coarser upper bound that tries to maintain
relationship among variables and error terms.  For example, the choices
of~\cite{DBLP:conf/cav/GhorbalGP09} try to compute such bounds while preserving
as much as possible the error terms of the operands, providing a precise way to
approximate a functional.

\vspace{-0,2cm}

\paragraph{Related works.}
Zonotopes are mainly used in static analysis to support the formal verification
of critical systems performing floating point computation, e.g. aircraft
controllers. One can mention a first line of works in which
zonotopes are used to precisely over-approximate set of values:
\begin{inparaenum}
\item hybrid system simulation, for example set-based simulation~\cite{DBLP:conf/rsp/BouissouMC12};
%\item loop unrolling
\item or floating point error propagation~\cite{DBLP:conf/sas/Goubault13}.
\end{inparaenum}
In those cases, a join operator is not necessarily needed nor a partial order check.

A second line of work tries to rely on this representation to perform classical
abstract interpretation. Zonotopes are then fitted with a computable partial order and a
join %Because lub do not exist in general for zonotopes, alternatives upper
%bounds are used, sometimes they compute mub, aka minimal upper bounds.
e.g.~\cite{DBLP:journals/entcs/GoubaultGP12,DBLP:conf/sas/GoubaultPV12}. The
approach of~\cite{DBLP:conf/cav/GhorbalGP09} is available in the open-source
library APRON~\cite{DBLP:conf/cav/JeannetM09}.

% \begin{itemize}
% \item fluctuat~\cite{DBLP:conf/sas/Goubault13}
% \item t1p~\cite{DBLP:conf/cav/GhorbalGP09}
% \end{itemize}
%   Again in those works, since the join is not computing a lub, ie. computing too
%   coarse results,
%   lopp unrolling is used to delay as much as possible the use of the join
%   operator.

% Therefore one can divide the ``features'' or characterizing elements of
%   zonotopes in two sets:
% \begin{itemize}
% \item static analysis elements
% \begin{itemize}
% \item injection of intervals
% \item update by an assignment defined as an arithmetic expression
% \item constrain by a condition
% \end{itemize}
% \item abstract domain operators: poset structure
% \begin{itemize}
% \item partial order
% \item join (upper bound or minimal upper bound instead of lub)
% \end{itemize}
% \end{itemize}

Back in the applied mathematics community, variants of affine arithmetics have
been studied in~\cite{DBLP:journals/rc/MessineT06} among which the quadratic
extension of affine forms allowing to express terms in $\epsilon_i\epsilon_j$.

\vspace{-0,2cm}

\paragraph{Contributions.}
In the paper, we ambition at using zonotopes based on this quadratic
arithmetics. 
% The contributions
% %presented in this paper
%  are the following
We propose an abstraction based on an extension of zonotopic abstract domains
to quadratic arithmetic. Our approach fully handles floating point computations
and performs the necessary rounding to obtain a sound result. Furthermore, while
keeping the complexity reasonable, i.e. quadratic instead of linear in the
error terms, quadratic forms are best suited to represent non linear computations
such as multiplication. Interestingly, the geometric concretization a set of
quadratic forms characterizes a non convex, non symmetric subset of
$\mathbb{R}^n$, while still being fitted with an algebraic structure.

\vspace{-0,2cm}

\paragraph{Paper structure.} A first section presents quadratic forms as
introduced in ~\cite{DBLP:journals/rc/MessineT06}. Then Sec.~\ref{sec-3} presents our
extension of zonotopes to quadratic arithmetics. Sec.~\ref{sec:floats} motivates
our floating point implementation. Sec.~\ref{sec-5} proposes a more precise way
to project quadratic zonotopes to intervals using semi-definite programming (SDP) solvers. Finally
Sec.~\ref{sec:exp} addresses our implementation and the evaluation of the approach
with respect to existing domains (intervals, affine zonotopes variants).
\begin{comment}
The proposed approach enables the analysis of floating c

\begin{itemize}
\item we exhibited a Galois connection between sets of values represented by a
  quadratic form and sets of values represented by an affine form

  \emph{Result 1:} A quadratic form is more precise than an affine form

\item we lifted this relationship to zonotopic extension of those arithmetics:
  Affine Zonotopes are an abstraction of Quadratic Zonotopes

  \emph{Result 2:} we provide sound join and leq operators to quadratic
  zonotopes

\item we improved the projection of quadratic zonotopes to concrete values, ie
  intervals, using semi-definite programming

  \emph{Result 3:} Improved projection to intervals
\item we implemented the proposed algorithms and applied them to classical
  examples.

  \emph{Result 4:} efficient and available implementation in OCaml
\end{itemize}
\end{comment}
%%% Local Variables: 
%%% mode: latex
%%% TeX-master: "submission"
%%% End: 

%\input{stateofart}

\section{Formal Preliminaries: Quadratic forms}

We formally introduce here some definitions
from~\cite{DBLP:journals/rc/MessineT06} defining quadratic forms. We refer the
interested reader to this publication for a wider comparison in a global
optimization setting.

\subsubsection{Quadratic forms.}
A (not so) recent extension of affine arithmetics is quadratic arithmetics~\cite{DBLP:journals/rc/MessineT06}.  
It is a comparable representation of values fitted with similar arithmetics operators but quadratic
forms also considers products of two errors terms, i.e. in $\epsilon_i \epsilon_j$. A quadratic form
is also parametrized by additional error terms used to encode non linear errors:
$\epsilon_\pm \in [-1,1], \epsilon_+ \in [0,1]$ and $\epsilon_- \in [-1, 0]$.
Let us define the set $\ccm\triangleq[-1,1]^m\times [-1,1]\times [0,1]\times [-1,0]$.
A quadratic form on $m$ noise symbols is a function $q$ from $\ccm$ to $\rr$ defined for all $t=(\epsilon,\epsilon_{\pm},\epsilon_+,\epsilon_-)\in\ccm$by $q(t)=c + b^\intercal \epsilon +\epsilon^\intercal A \epsilon + c_\pm \epsilon_\pm + c_- \epsilon_- +c_+ \epsilon_+$. A quadratic form is thus characterized by a $6$-tuple $\left( c , (b)_m, (A)_{m^2}, c_\pm, c_+, c_- \right) \in\mathbb{R} \times \mathbb{R}^m \times \mathbb{R}^{m \times m} \times \mathbb{R}_+\times \mathbb{R}_+ \times \mathbb{R}_+$. To simplify, we will use the terminology quadratic form for both the function defined on $\ccm$ and the $6$-tuple. 
%$\left( c , (b)_m, (A)_{m^2}, c_\pm, c_+, c_- \right)$. 
We denote by $\Q^m$ the set of quadratic forms.
\subsubsection{Geometric interpretation.}
Let $q\in\Q^m$. Since $q$ is continuous, the image of $\ccm$ by $q$ is a closed bounded interval. In our context, the image of $\ccm$ by $q$ defines its geometric interpretation. 
%concretization function.
\begin{definition}[Concretization of quadratic forms]
  \label{def:gammaQ}
The concretization map of a quadratic form $\gamma_{\Q}:  \Q^m  \rightarrow  \wp (\mathbb{R})$ is defined by:
%The concretization map $\gamma_{\Q}$ is defined, for all $q\in\Q^m$ by:
%\[
%\begin{array}{l}
  %\gamma_{\Q^m} (c,(b)_m, (A)_{m^2},c_\pm, c_+, c_-) = \\
  %\qquad \left\{ x \in \mathbb{R} \left|
      %\begin{array}{l}
        %\exists \epsilon \in
        %[-1;1]^m, \epsilon_\pm \in [-1;1], \epsilon_+ \in
        %[0;1], \epsilon_- \in [-1, 0]\\ \textrm{ s.t. } x = c + b^\intercal \epsilon +
        %\epsilon^\intercal A \epsilon + c_\pm \epsilon_\pm + c_- \epsilon_- +
        %c_+ \epsilon_+
%\end{array}
      %\right\}\right.
 %\end{array}
 %\]
 \[
  \gamma_{\Q} (q) = \left\{ x \in \mathbb{R} \left|\exists\, t \in\ccm\ \st\ x = q(t)\right\}\right.
 \]
\end{definition}
\begin{remark}
\label{notorder}
We can have $\gamma_{\Q}(q)=\gamma_{\Q}(q')$ with $q\neq q'$ e.g. $q=\epsilon_1^2$ and $q'=\epsilon_2^2$. 
\end{remark}
%We could consider equivalence classes instead to get an order but we would loose the information that $q_1$ and $q_2$ are not correlated.

The concretization of $q$ consists in computing the infimum and the supremum of $q$ over $\ccm$ i.e. the values: 
\begin{equation}
\label{exactbounds}
\mathbf{b}^q\triangleq\inf\{q(x)\mid x\in \ccm\} \qquad \text{and}\qquad \mathbf{B}^q\triangleq \sup\{q(x)\mid x\in \ccm\}\enspace.
\end{equation}
To compute $\mathbf{b}^q$ and $\mathbf{B}^q$ is reduced to solve a non-convex quadratic problem which is NP-hard~\cite{Vavasis199073}. 
The approach described in~\cite{DBLP:journals/rc/MessineT06} uses simple inequalities to give a safe over-approximation of $\gamma_\Q(q)$. The interval provided by this approach is $[\mathbf{b}^q_{MT},\mathbf{B}^q_{MT}]$ defined as follows:
\begin{equation}
%\label{touamimessinebounds}
%\left\{
%\begin{array}{c}
%\mathbf{b}^q_{MT}\triangleq\displaystyle{c - \sum_{i=1}^{m} |b_i| - \sum_{i=1}^{m} \sum_{\substack{j=1,\ldots,m\\j\neq i}} |A_{ij} |+\sum_{i=1}^{m} [A_{ii}]^- -c_- - c_{\pm}}\\
%\mathbf{B}^q_{MT}\triangleq\displaystyle{c +\sum_{i=1}^{m} |b_i| + \sum_{i=1}^{m} \sum_{\substack{j=1,\ldots,m\\j\neq i}} | A_{ij} | +\sum_{i=1}^{m} [A_{ii}]^+ +c_+ + c_{\pm}}
%\end{array}
%\right.
\label{touamimessinebounds}
\left\{
\begin{array}{c}
\mathbf{b}^q_{MT}\triangleq\displaystyle{c - \sum_{i=1}^{m} |b_i| - \sum_{\substack{i,j=1,\ldots,m\\j\neq i}} |A_{ij} |+\sum_{i=1}^{m} [A_{ii}]^- -c_- - c_{\pm}}\\
\mathbf{B}^q_{MT}\triangleq\displaystyle{c +\sum_{i=1}^{m} |b_i| + \sum_{\substack{i,j=1,\ldots,m\\j\neq i}} | A_{ij} | +\sum_{i=1}^{m} [A_{ii}]^+ +c_+ + c_{\pm}}
\end{array}
\right.
\end{equation}
where for all $x\in\rr$,
$[x]^+=x$ if $x>0$ and 0 otherwise and $[x]^-=x$ if $x<0$ and 0 otherwise.

In practice, we use $\gamma_{\Q}^{MT}(q)\triangleq[\mathbf{b}^q_{MT},\mathbf{B}^q_{MT}]$ instead of $\gamma_\Q(q)$. In Sec.~\ref{sec-5}, we will present a tighter safe over-approximation of $\gamma_\Q(q)$ using SDP. 

We will need a "reverse" map to the concretization map $\gamma_{\Q}$: a map which associates to an interval a quadratic form. We call this map the \emph{abstraction map}. Note that the abstraction map produces a fresh noise symbol. 

First, we introduce some notations for intervals. 
Let $\INT$ be the set of closed bounded real intervals i.e. $\{[a,b]\mid
a,b\in\rr, a \leq b\}$ and $\bINT$ its unbounded extension, i.e. $a \in \rr \cup
\{-\infty\}, b \in \rr \cup
\{+\infty\}$. $\forall [a,b]\in \INT$, we define two functions $\operatorname{lg}([a,b]) = (b-a)/2$
and $\midop([a,b])=(b+a)/2$. 
% We also need to the whole set of real intervals
% $\bINT\triangleq \INT\cup\{[a,+\infty[, a\in\rr\}\cup \{]-\infty,a], a\in\rr\}$.
Let $\sint$ be the classical join of $\INT$ that is $[a,b] \sint
[c,d]\triangleq [\min(a,c),\max(b,d)]$. Let $\sqcap_{\bINT}$ be the classical meet
of intervals.
 \begin{definition}[Abstraction]
  \label{def:abstractQ}
The abstraction map $\alpha_{\Q}:  \INT  \rightarrow  \Q^1$ is defined by:
 \[
 \alpha_{\Q} ([a_1,a_2]) =(c, (b)_{1},(0)_1, 0, 0, 0) \text{ where } c=\midop\left([a_1,a_2]\right) \text{ and } b= \operatorname{lg} \left([a_1,a_2]\right)\enspace .
\]
\end{definition}
\begin{property}[Concretization of abstraction]
\label{concabst}
$\gamma_{\Q}\left(\alpha_\Q\left([a_1,a_2]\right)\right)=[a_1,a_2]$.
\end{property}
\subsubsection{Arithmetic operators.}
Quadratic forms are fitted with arithmetic operators which complexity is quadratic in
the number of error terms.  We give here the definitions of the arithmetics
operators:
\begin{definition}[Arithmetics operator in $\Q$]
\label{def:arithq}
Addition, negation, multiplication by scalar are defined by:
  \[
  \begin{array}{l}
    (c,(b)_m, (A)_{m^2},c_\pm, c_+, c_-) +_\Q (c',(b')_m, (A')_{m^2},c'_\pm,
    c'_+, c'_-) =\\
    \hfill(c+c', (b+b')_m, (A+A')_{m^2}, c_\pm + c'_\pm, c_+ + c'_+,
    c_- + c'_-)\\
    -_\Q (c,(b)_m, (A)_{m^2},c_\pm, c_+, c_-) = (-c,(-b)_m, (-A)_{m^2},c_\pm,
    c_-, c_+) \\
    \lambda *_\Q (c,(b)_m, (A)_{m^2},c_\pm, c_+, c_-) = (\lambda
    c,\lambda(b)_m, \lambda(A)_{m^2},|\lambda| c_\pm, |\lambda| c_+, |\lambda| c_-)\\
  \end{array}
\]
%\assale{Attention pour la formule en $\lambda *_\Q$, on devrait avoir $|\lambda| c_\pm, |\lambda| c_+, |\lambda| c_-$
%pour conserver les contraintes de positivité.}
\noindent The multiplication is more complex since it introduces additional errors.
\[
  \begin{array}{l}

    (c,(b)_m, (A)_{m^2},c_\pm, c_+, c_-) \times_\Q (c',(b')_m,
    (A')_{m^2},c'_\pm, c'_+, c'_-) = \\
    \hfill
    \left\{
      \begin{array}{l}
      (c c', c'(b)_m, + c(b')_m, c'(A)_{m^2}+ c(A')_{m^2} +
      (b)_m(b')_m^\intercal, c''_\pm, c''_+, c''_- \textrm{ with}\\
      c''_x = c''_{x_1}  + c''_{x_2} + c''_{x_3} + c''_{x_4}, \forall x \in \{+, -, \pm\}
    \end{array}
    \right.
  \end{array}
  \]
Each $c''_{x_i}$ accounts for multiplicative errors with more than quadratic
degree, obtained in the following four
sub terms:
\begin{inparaenum}[(1)]
\item $\epsilon^\intercal A \epsilon \times \epsilon^\intercal A' \epsilon$
\item $b^\intercal \epsilon \times \epsilon^\intercal A' \epsilon$ and $b'^\intercal \epsilon \times \epsilon^\intercal A \epsilon$
\item multiplication of a matrix element in $A$, $A'$ times an error term in
  $\pm, +, -$
\item multiplication between error terms or with constant $c$, $c'$.
\end{inparaenum}
Their precise definition can be found in~\cite[\textsection 3]{DBLP:journals/rc/MessineT06}.
\end{definition}

\section{Quadratic Zonotopes: a zonotopic extension of quadratic forms to environments}
\label{sec-3}
%\assale{quadratic vectors...fait penser à des ensembles...quadratic vectors/ tuple, c'est pas mal non?}

% \ploc{j'ai recopié le code des affines, A MODIFIER}

% poset

% join

% operator

% soundness

%\subsection{Zonotope: extension of forms to environnements.}
\label{sec-1-3-1-3}%
Quadratic vectors are
 the 
lift to
 environments
 of quadratic forms. They provide a
p-dimensional environment in which each dimension/variable is associated to a
quadratic form. As for the affine sets used in zonotopic
domains~\cite{DBLP:journals/corr/abs-0910-1763}, the different variables share (some) error terms, this
characterizes a set of relationships between variables, when varying the values
of $\epsilon$ within $[-1,1]^m$. The geometric interpretation of quadratic vectors
are non convex non symmetric subsets of $\mathbb{R}^p$. In the current paper, we
call them Quadratic Zonotopes to preserve the analogy with affine sets and
zonotopes. 

\begin{example}[quadratic vector]
\label{ex:quadzon}
  Let us consider the following quadratic vector $q$:
\[q = \begin{cases} 
    x = -1 + \epsilon_1 - \epsilon_2 - \epsilon_{1,1} \\
    y = 1 + 2 \epsilon_2 + \epsilon_{1,2}
  \end{cases}
\]

 Fig.~\ref{fig:conc_quad_zon} represents its associated geometric
 interpretation, a quadratic zonotope.
\end{example}

% pour le moment figure simple on verra apres
\begin{wrapfigure}{r}{.5\textwidth}
  \vspace{-1cm}
%\begin {figure}
    \centering
    \resizebox{.45\textwidth}{!}{\begin{tikzpicture}
\draw [black!50!cyan,fill= white!90!cyan] (-2,0) .. controls (-0.8,-0.45) .. (-0.4,-0.7) -- (0,-2) .. controls (0.3,-1.6) .. (0,-0.7) -- (-1.75,3.6) .. controls (-1.85,3.8) .. (-2,4) -- (-1.8,3.4) .. controls (-2,2.85) .. (-4,2) -- cycle;

\draw [->] (-5,0) -- (3,0) node[anchor=north] {$x$} ;
\draw [->] (0,-3) -- (0,5) node[anchor=east] {$y$};
\foreach \x in {-4,-2,0,2} {
\draw (\x,-1pt) -- (\x,1pt)  node[anchor=north east] {\x};
}
\foreach \y in {-2,4} {
\draw (-1pt,\y) -- (1pt,\y) node[anchor=north east] {\y};
}
\end{tikzpicture}}
%\vspace{-1em}
    \caption{Zonotopic concretization of the quadratic vector $q \in \ZQ{p}{m}$ of Ex.~\ref{ex:quadzon}: $\gamma_{\ZQ{}{}}(q)$\vspace{-1.8em}}
    \label{fig:conc_quad_zon}
  \end{wrapfigure}%
%  \end{figure}%
We denote by $\ZQ{p}{m}$ such quadratic vectors of dimension $p$:
$(q^p) \in \ZQ{p}{m} = \left( c^p , (b)^p_m, (A)^p_{m^2}, c^p_\pm, c^p_+, c^p_- \right) \in
\mathbb{R}^p \times \mathbb{R}^{p \times m} \times \mathbb{R}^{p \times m \times m} \times \mathbb{R}^p_+
\times \mathbb{R}^p_+ \times \mathbb{R}^p_+$.

The Zonotope domain is then a parametric relational abstract domain,
parametrized by the vector of $m$ error terms. In practice, its definition mimics a
non relational domain based on an abstraction $\ZQ{p}{m}$ of $\wp(\mathbb{R}^p)$.
Operators are \begin{inparaenum}[(i)] 
\item assignment of a variable of
the zonotope to a new value defined by an arithmetic expression, using the
semantics evaluation of expressions in $\Q$ and the substitution in the quadratic vector;
\item guard evaluation, i.e. constraint over a zonotope, using the classical
  combination of forward and backward evaluations of
  expressions~\cite[\textsection 2.4.4]{minephd}.
\end{inparaenum}

%The drawback of this non relational domain structure %, as mentionned in
%%Sec.~\ref{sec:soa}, 
%is the imprecision of the guard evaluation since it relies
%on the use of the meet operator $\sqcap_\Q$ which looses all relationships to
%existing error terms and bounds a fresh one.

% \begin{remark}
%   When performing static analysis using zonotopes, for example using unrolling,
%   one can increase widely the set of error terms since the complexity of the
%   operators is linear in the number of error terms.
% \end{remark}

\subsubsection{Geometric interpretation and box projection.}
%  \label{sec-1-3-1-3-1}%
One can consider the geometric interpretation as the concretization of a
quadratic vector to a quadratic zonotope. 

From now on, for all $n\in\mathbb{N}$, $[n]$ denotes the set of integers $\{1,\ldots,n\}$.
 
\begin{definition}[Concretization in $\ZQ{p}{m}$]
  %A concretization operator from quadratic vectors to subsets of $\mathbb{R}^p$ is then
  %defined as
  %\[\begin{array}{rcl}
    %\label{def:gammaZQ}
    %\gamma_{\ZQ{}{}}: \ZQ{p}{m} &
    %\rightarrow& 
    %\wp (\mathbb{R}^p)\\
    %\left( c^p , (b)^p_m, (A)^p_{m^2}, c^p_\pm, c^p_+, c^p_- \right) &
    %\mapsto & \\&& 
    %\hspace{-1.5cm}\left\{ x \in \mathbb{R}^p \left|
        %\begin{array}{l}
%\exists \epsilon \in [-1;1]^m, \epsilon_\pm
      %\in [-1; 1], \epsilon_+ \in [0; 1], \epsilon_- \in [-1; 0],\\ \textrm{ s.t.
      %}x = c + b^\intercal \epsilon + \epsilon^\intercal A \epsilon + c_\pm 
      %\epsilon_\pm + c_+ \epsilon_+ + c_- \epsilon_-
    %\end{array}\right\}\right.\end{array}
  %\]
    The concretization map $\gamma_{\ZQ{}{}}:\ZQ{p}{m}\mapsto \wp\left(\rr^p\right)$ is defined for all $q=(q_1,\ldots,q_p)\in\ZQ{p}{m}$ by:
  \[
    \gamma_{\ZQ{}{}}\left( q\right)=\left\{ x \in \mathbb{R}^p \left|
\exists\, t \in\ccm\ \st\ \forall\, k\in [p],\ x_k = q_k(t)\right\}\right. \enspace .
  \]
\end{definition}
\begin{remark}
  Characterizing explicitly such subset of $\mathbb{R}^p$ as a set of
  constraint is not easy. A classical (affine) zonotope is the image of a polyhedron (hypercube)
  by an affine map, hence it is a polyhedron and can be represented by a
  conjunction of affine inequalities. In the quadratic vectoring, such
  representation as conjunction of quadratic or at most polynomial inequalities
  is not proven to exist. This makes the concretization of a quadratic
  set difficult to compute precisely.
\end{remark}

To ease the later interpretation of computed values, we rely on a naive
projection to boxes: each quadratic form of the quadratic vector is concretized as
an interval using $\gamma_\Q$.

\subsubsection{Preorder structure.}
%  \label{sec-1-3-1-3-2}%
%We can equip quadratic vectors with a partial order relying on the geometric inclusion provided by
%Recall that a preorder $\preceq$ is a binary relation that is reflexive $x\preceq x$ and transitive 
%$x\preceq y\land y\preceq z$ implies that $x\preceq z$. 
We can fit quadratic vectors with a preorder relying on the geometric inclusion provided by the map $\gamma_{\mathcal{Z}_{\Q}}$.
%  We define the partial order over zonotopes as the geometric inclusion using
%the $\gamma_{\mathcal{Z}_{\Q}}$ map.

%\begin{definition}[Partial order in $\ZQ{p}{m}$]
%The partial order  $\sqsubseteq_{\ZQ{}{}}$ over $\ZQ{p}{m}$ is defined as
  %\[\begin{array}{rcl}
    %\label{def:orderZQ}
    %\sqsubseteq_{\ZQ{}{}}: \ZQ{p}{m} \times \ZQ{p}{m} &
    %\rightarrow& 
    %\mathbb{B}\\
    %x \sqsubseteq_{\ZQ{}{}} y &
    %\mapsto & \gamma_{\ZQ{}{}}(x) \subseteq
    %\gamma_{\ZQ{}{}}(y)
  %\end{array}\]
%\end{definition}

%\begin{definition}[Partial order in $\ZQ{p}{m}$]
\begin{definition}[Preorder in $\ZQ{p}{m}$]
%The partial order  $\sqsubseteq_{\ZQ{}{}}$ over $\ZQ{p}{m}$ is defined by:
The preorder  $\sqsubseteq_{\ZQ{}{}}$ over $\ZQ{p}{m}$ is defined by:

  \[
    x \sqsubseteq_{\ZQ{}{}} y \iff \gamma_{\ZQ{}{}}(x) \subseteq
    \gamma_{\ZQ{}{}}(y)\enspace .
 \]
\end{definition}

\begin{remark}
  %Note that this partial order operator is not computable in general. 
Since $\gamma_{\ZQ{}{}}$ is not computable, $x\sqsubseteq_{\ZQ{}{}}y$ is not decidable. Note also that, from Remark~\ref{notorder}, the binary relation $\sqsubseteq_{\ZQ{}{}}$ cannot be antisymmetric and thus cannot be an order.
\end{remark}
%Then we can only equip $\left(\ZQ{p}{m}, \sqsubseteq_{\mathcal{Z}_\Q}\right)$ with a preorder structure and not a lattice structure.
%The binary relation $\sqsubseteq_{\ZQ{}{}}$ cannot be antisymmetric. 
%Indeed, let us take $q^1=\epsilon_1^2$ and $q^2=\epsilon_2^2$. Then, $\gamma_{\ZQ{}{}}(q^1)=\gamma_{\ZQ{}{}}(q^2)$ but $q^1\neq q^2$. 
%We could consider equivalence classes instead to get an order but we would loose the information that $q_1$ and $q_2$ are not correlated.   
\begin{remark}
The least upper bound of $Z \subseteq\ZQ{p}{m}$ i.e. an element $z'$ s.t.
%Z = $\{ z \in\ZQ{p}{m} \}$, an element $z'$ s.t.
%\begin{align*} \forall z \in Z, z \sqsubseteq_{\mathcal{Z}_{\Q}}
  %z' \\
 %\forall z'' \in \ZQ{p}{m}, \forall z \in Z, z
  %\sqsubseteq_{\mathcal{Z}_\Q} z'', \textrm{ then } z'
  %\sqsubseteq_{\mathcal{Z}_\Q} z'' \end{align*} 

$\left(\forall z \in Z, z \sqsubseteq_{\mathcal{Z}_{\Q}} z' \land
 \forall z'' \in \ZQ{p}{m}, \forall z \in Z, z\sqsubseteq_{\mathcal{Z}_\Q} z''\right) \implies z'\sqsubseteq_{\mathcal{Z}_\Q} z''$  
does not necessarily exists.
\end{remark}
Related
work~\cite{DBLP:journals/corr/abs-0910-1763,DBLP:conf/cav/GhorbalGP09,DBLP:conf/cav/GhorbalGP10,DBLP:journals/entcs/GoubaultGP12,DBLP:conf/sas/GoubaultPV12}
addressed this issue by providing various flavors of join operator computing a
safe upper bound or a minimal upper bound. Classical Kleene iteration scheme was
adapted to fit this loose framework without least upper bound computation. Note that, in general, the
aforementioned zonotopic domains do not rely on the geometric interpretation as
the concretization to $\wp(\mathbb{R})$.

We detail here the join operator considered in this paper. It is the lift of the
operator proposed in \cite{DBLP:journals/corr/abs-0910-1763} to quadratic
vectors. The motivation of this operator is to provide an upper bound while
minimizing the set of error terms lost in the computation. 

First we introduce a useful function $\argmin$: it cancels values of
opposite sign but provides the argument with the minimal absolute value when
provided with two values of the same sign:
\begin{definition}[Argmin] We define for all $a\in\rr$, $\operatorname{sgn}(a)=1$ if $a\geq 0$ and -1 otherwise. 
The $\argmin$ function, $\argmin: \mathbb{R} \times \mathbb{R} \rightarrow \mathbb{R}$ is defined as:
%\[
%\forall a,b \in \mathbb{R}, 
%\argmin(a,b) = 
%\begin{cases}
%\min(a,b) & if \ a \geq 0 \ and \ b \geq 0\\
%\max(a,b) & if \ a \leq 0 \ and \ b \leq 0\\
%0 & if \ ab \leq 0\\
%\end{cases}
%\]
$\forall a,b \in \mathbb{R}, \argmin(a,b) = \operatorname{sgn}(a)\min(|a|,|b|)$ if $ab \geq 0$ and 
0 otherwise.
\end{definition}

%\begin{definition} The projection function $\pi_k: \ZQ{p}{m} \rightarrow \Q^m$ provides the quadratic form describing the $k$-th component of the quadratic vector. It is defined as
%$\forall ( c^p , (b)^p_m, (A)^p_{m^2}, c^p_\pm,$ $c^p_+, c^p_- ) \in
%\ZQ{p}{m}, k \in [1;p],$ \[
%\begin{array}{l}
 
%\pi_k \left(c^p , (b)^p_m, (A)^p_{m^2}, c^p_\pm, c^p_+, c^p_-\right) = \left((c^p)_k , ((b)_m^p)_k, ((A)_m^p)_k, (c^p_\pm)_k, (c^p_+)_k, (c^p_-)_k\right)
%\end{array}
%\]
%\end{definition}
We also need the projection map which selects a specific coordinate of a quadratic vector.
\begin{definition}[Projection] 
%The projection map $\pi_k: \ZQ{p}{m} \rightarrow \Q^m$ provides the $k$-th quadratic form describing the $k$-th component of the quadratic vector. It is defined as: $\forall\, q=\left(q_1,\ldots,q_p\right) \in\ZQ{p}{m}, k \in [p]$, $\pi_k(q) = q_k$.
The projection map $\pi_k: \ZQ{p}{m} \rightarrow \Q^m$ is defined by: $\forall\, q=\left(q_1,\ldots,q_p\right) \in\ZQ{p}{m},\ \forall\, k \in [p]$, $\pi_k(q) = q_k$.
\end{definition}
%\begin{definition} Let $n,m\in\mathbb{N}$ such that $n>m$. The extention function $\ext_n: \Q^{m} \rightarrow \Q^n$ is defined as: $\forall\, q=(c,(b)_m,(A)_m^2,c_{\pm},c_+,c_-)\in\Q^{m}, \ext_n(q) =(c,(b')_n,(A')_n^2,c_{\pm},c_+,c_-)\in\Q^{n}$ where $b'=(b,0_{n-m})^\intercal$ and $A_{i,j}'=A_{i,j}$ if $i,j\leq m$ and 0 otherwise.
%\end{definition}
%We also use the notation $\ext_n$ for a quadratic vector and for all $q\in\ZQ{p}{m}$, $\ext_n(q_1,\ldots, q_p)=(\ext_n(q_1),\ldots,\ext_n(q_p))$. 
%Let $n,m\in\mathbb{N}\ \st\ n>m$. We readily have for all $t=(\epsilon,\epsilon_{\pm},\epsilon_+,\epsilon_-)\in\ccm$ for all $t'=(\epsilon', \epsilon_{\pm},\epsilon_+,\epsilon_-)\in\ccn$, such that for all $k\in [m]$, $\epsilon'_k=\epsilon_k$,$q(t)=\ext_n(q)(t')$ and thus $\gamma_{\ZQ{}{}}(q)=\gamma_{\ZQ{}{}}(\ext_n(q))$.
When a quadratic form $q$ is defined before a new noise symbol creation, we have to extend $q$ to take into account this fresh noise symbol.
\begin{definition}[Extension] Let $i,j\in\mathbb{N}$. The extension map $\ext_{i,j}:\Q^{m} \rightarrow \Q^{i+j+m}$ is defined by: $\forall\, q=(c,(b)_m,(A)_m^2,c_{\pm},c_+,c_-)\in\Q^{m}, \ext_{i,j}(q) =(c,(b')_{i+j+m},(A')_{(i+j+m)^2},c_{\pm},c_+,c_-)\in\Q^{n}$ where $b_k'=b_{k-i}$ if $i+1\leq k\leq m+i$ and 0 otherwise and $A_{k,l}'=A_{k-i,l-i}$ if $i+1\leq k,l\leq m+i$ and 0 otherwise.
\end{definition}
\begin{property}[Extension properties]
\label{extensionprop}
Let $i,j\in\mathbb{N}$.
\begin{enumerate}
\item Let $t=(\epsilon,\epsilon_\pm,\epsilon_+,\epsilon_-)\in\ccm$ and $t'=(\epsilon',\epsilon_\pm,\epsilon_+,\epsilon_-)\in \cc^{m+i+j} \st \forall\, i+1\leq k \leq m+i$, $\epsilon'_k=\epsilon_{k-i}$. Then 
$q(\epsilon,\epsilon_{\pm},\epsilon_+,\epsilon_-)=\ext_{i,j}(q)(\epsilon',\epsilon_{\pm},\epsilon_+,\epsilon_-)$.
%Let $q\in\Q^m$. 
%For all $\epsilon\in [-1,1]^m$ and for all $\epsilon'\in [-1,1]^{m+i+j}$ such that $\epsilon'_k=\epsilon_{k-i}$ for all $i+1\leq k \leq m+i$, for all $(\epsilon_{\pm},\epsilon_+,\epsilon_-)\in [-1,1]\times [0,1]\times [-1,0]$ $q(\epsilon,\epsilon_{\pm},\epsilon_+,\epsilon_-)=\ext_{i,j}(q)(\epsilon',\epsilon_{\pm},\epsilon_+,\epsilon_-)$.
\item For all $q \in\Q^m$, $\gamma_\Q(q)=\gamma_\Q(\ext_{i,j}(q))$.
\end{enumerate}
\end{property}
Now, we can give a formal definition of the upper bound of two quadratic vectors.  
\begin{definition}[$\sqcup_{\ZQ{}{}}$: Upper bound computation in $\ZQ{p}{m}$]
The upper bound $\sqcup_{\ZQ{}{}}: \ZQ{p}{m} \times \ZQ{p}{m} \rightarrow \ZQ{p}{m+p}$ is defined, 
for all $q=\left(c, b, A, c_\pm, c_+, c_-\right), q'=\left(c', b', A', c'_\pm, c'_+, c'_-\right)\in\ZQ{p}{m}$ by:
\[ 
q \sqcup_{\ZQ{}{}} q' = \left(\ext_{0,p}(q_k'')\right)_{k\in [p]} + q^{e} \in \ZQ{p}{m+p}
\]
where $q''=(c'', (b'')_m^p, (A'')_{m^2}^p, c''^p_\pm, c''^p_+, c''^p_-)\in \ZQ{p}{m}$
%\rr^p\times \rr^{p \times m}\times \rr^{p \times m \times m}\times \rr^{p}\times \rr^{p}\times \rr^{p}$ 
with, for all $k\in[p]$:
\begin{itemize}
\item $(c'')_k = \midop (\gamma_\Q ( \pi_k(q)) \cup \gamma_\Q ( \pi_k(q')));$
\item $\forall\, t \in \{ \pm, +, - \}, c''_{t,k} = \argmin (c_{t,k},c'_{t,k});$
\item $\forall\, i\in [m],\ (b'')_{k,i} = \argmin (b_{k,i},b'_{k,i});$
\item $\forall\, i,j\in [m],\ (A'')_{k,i,j} = \argmin(A_{k,i, j},A'_{k,i, j});$
\end{itemize}
%and$q^{e} = ((c^{e})^p, (b^{e})_{(m+p)}^{p}, (0)^p, 0^p, 0^p, 0^p)$ with
and $\forall\, k \in[p]$,
%\begin{itemize}
%\item $c^{e}_k=\midop\left(\gamma_\Q ( \pi_k(q) - \pi_k(q'')) \cup \gamma_\Q ( \pi_k(q') - \pi_k(q'')) \right)$,
%\item ${b^{e}}_{k, m+k'}= 
%\begin{cases}
%\operatorname{lg} \left(\gamma_\Q ( \pi_k(q) - \pi_k(q'')) \cup \gamma_\Q ( \pi_k(q') - \pi_k(q'')) \right) \ \text{if} \ k' = k\\
%0 \ \text{otherwise}
%\end{cases}
%$
%\end{itemize}
$q_k^{e} = \ext_{(m+k-1),(p-k)}\left(\alpha_{\Q}\left(C_k \sint C'_k\right)\right)\text{ with }C_k=\gamma_\Q ( \pi_k(q) - \pi_k(q''))\text{ and }C'_k=\gamma_\Q ( \pi_k(q') - \pi_k(q''))\enspace .$
\end{definition}
%for $Z \subseteq\ZQ{p}{m}$, an element $z'$ s.t.
%Z = $\{ z \in\ZQ{p}{m} \}$, an element $z'$ s.t.
%\begin{align*} \forall z \in Z, z \sqsubseteq_{\mathcal{Z}_{\Q}}
  %z' \\
 %\forall z'' \in \ZQ{p}{m}, \forall z \in Z, z
  %\sqsubseteq_{\mathcal{Z}_\Q} z'', \textrm{ then } z'
  %\sqsubseteq_{\mathcal{Z}_\Q} z'' \end{align*} 
%\[
%\left(\forall z \in Z, z \sqsubseteq_{\mathcal{Z}_{\Q}} z' \land
% \forall z'' \in \ZQ{p}{m}, \forall z \in Z, z\sqsubseteq_{\mathcal{Z}_\Q} z''\right) \implies z'\sqsubseteq_{\mathcal{Z}_\Q} z'' 
%\]  
%does not necessarily exists.

Let us denote the Minkowski sum and the Cartesian product of sets by respectively $D_1\oplus D_2=\{d_1+d_2\mid d_1\in D_1,\ d_2\in D_2\}$ and $\prod_i^n D_i=\{(d_1,\ldots,d_n)\mid\forall\, i\in [n],\ d_i\in D_i\}$. We have the nice characterization of the concretization of the upper bound given by Lemma~\ref{lemmajoin}. 
\begin{lemma}
\label{lemmajoin}
%By construction $q''$ and $q^{err}$ satisfy: 
%\begin{equation}
%\label{eqlemmajoin}
%\gamma_{\ZQ{}{}}(q''+q^{err})=\gamma_{\ZQ{}{}}(q'')\oplus \prod_{k=1}^p \gamma_{\Q^{m+p}}(q_k^{err})
%\end{equation}
 %where $\oplus$ denotes the Minkowski sum of sets (i.e. $A\oplus B=\{a+b\mid a\in A,\ b\in B\}$
%and $\prod$ is the Cartesian product of sets (i.e. $A\times B=\{(a,b)\mid a\in A,\ b\in B\}$. 
By construction of $q''$ and $q^{e}$ previously defined:
\[ 
\gamma_{\ZQ{}{}}\left(\left(\ext_{0,p}(q_k'')\right)_{k\in [p]}+q^{e}\right)=\gamma_{\ZQ{}{}}(q'')\oplus \displaystyle{\prod_{k=1}^p \gamma_{\Q^{m+p}}(q_k^{e})}
\]
\end{lemma}
\begin{proof}
%See Proof~\ref{prooflemmajoin}.
See Appendix.
\end{proof}
Now, we state at Theorem~\ref{thjoin} that the $\sqcup_{\ZQ{}{}}$ operator computes an upper bound of its operands with respect to the preorder $\sqsubseteq_{\ZQ{}{}}$.
\begin{theorem}[Soundness of the upper bound operator]
\label{thjoin}
For all $q,q'\in \ZQ{p}{m}$,
$q\sqsubseteq_{\ZQ{}{}} q \sqcup_{\ZQ{}{}} q'$ and $q'\sqsubseteq_{\ZQ{}{}} q \sqcup_{\ZQ{}{}} q'$.
\end{theorem}
\begin{proof}
%See Proof~\ref{proofthjoin}.
See Appendix.
\end{proof}
%\input{preuve_theorem1}
% mid value + argmin
% Cerr as nul value, b, A term but
% p new error terms
% b^k_{n+k'} = 0 si k != k'
% = sup(gamma pi k C union gamma pi k C') - sup (gamma pi k C'')

% \begin{itemize}
% \item $(c'')_k = mid (\gamma_\A ( \pi_k(c^p,(b)^p_n)) \cup \gamma_\A ( \pi_k(c'^p,(b')^p_n)))$;
% \item let $q \in \mathbb{R}^{p \times n}$, s.t. $(q)_{k,i} = argmin (b_{k,i},b'_{k,i})$;
% \item $(b'')_{k,i} = (q)_{k,i}$;
% \item $\forall k' \in [1,p] \ \ (b'')_{k,n+k'} =\\
% \begin{cases}
% sup \left(\gamma_\A ( \pi_k(c^p,(b)^p_n)) \cup \gamma_\A ( \pi_k(c'^p,(b')^p_n) \right) - sup(\gamma_\A (\pi_k (c''^p,(q)_n^p)
% %_{\substack{i \in [1,p] \\ j \in [1,n]}}, c^Z)
% )) \ \textrm{when} \ k = k'\\
% 0 \ \textrm{otherwise}
% \end{cases}$
% \end{itemize}

\begin{example}
Let $Q$ and $Q'$ be two quadratic vectors:
\[Q = \begin{cases}
  x = -1 + \epsilon_1 - \epsilon_2 - \epsilon_{1,1} \\
  y = 1 + 2 \epsilon_2 + \epsilon_{1,2}
\end{cases} \hfill
Q'= \begin{cases}
  x = - 2 \epsilon_2 - \epsilon_{1,1} + \epsilon_+ \\
  y = 1 + \epsilon_1 + \epsilon_2 + \epsilon_{1,2}
\end{cases}
\]
The resulted quadratic vector $Q'' = Q \sqcup_{\ZQ{}{}} Q'$ is
        \[Q'' = \begin{cases}
        x = - \epsilon_2 - \epsilon_{1,1} + 2 \epsilon_3\\
        y = 1 + \epsilon_2 + \epsilon_{1,2} + \epsilon_4
      \end{cases}\]
  
\end{example}

 \begin{figure}
    \centering
    \begin{subfigure}{.55\textwidth}
      \resizebox{!}{5cm}{\begin{tikzpicture}

\draw [black!50!green,fill= white!90!green] (-3,0) -- (-4,0) -- (-4,4) -- (0,4) -- (3,1) -- (3,-1) .. controls (3,-1) and (3,-1.5) .. (2,-2) -- (-2,-2) -- cycle node[anchor=north west] {Q"}; 

\draw [black!50!cyan,fill= white!90!cyan, dashed] (-2,0) .. controls (-0.8,-0.45) .. (-0.4,-0.7) -- (0,-2) .. controls (0.3,-1.6) .. (0,-0.7) -- (-1.75,3.6) .. controls (-1.85,3.8) .. (-2,4) -- (-1.8,3.4) .. controls (-2,2.85) .. (-4,2) -- cycle node[anchor=north east] {Q};

\draw [black!50!red,fill= white!90!red, dashed] (-4,4) -- (-3,4) -- (3,0) -- (-4,0) .. controls (-2,1) and (-2,2) .. (-2,2) -- cycle node[anchor=north west] {Q'}; 

\draw [black!50!cyan, dashed] (-2,0) .. controls (-0.8,-0.45) .. (-0.4,-0.7) -- (0,-2) .. controls (0.3,-1.6) .. (0,-0.7) -- (-1.75,3.6) .. controls (-1.85,3.8) .. (-2,4) -- (-1.8,3.4) .. controls (-2,2.85) .. (-4,2) -- cycle;

\draw [->] (-5,0) -- (4,0) node[anchor=north] {$x$} ;
\draw [->] (0,-3) -- (0,5) node[anchor=east] {$y$};
\foreach \x in {-4,-2,0,2,3} {
\draw (\x,-1pt) -- (\x,1pt)  node[anchor=north east] {\x};
}
\foreach \y in {-2,4} {
\draw (-1pt,\y) -- (1pt,\y) node[anchor=north east] {\y};
}
\end{tikzpicture}}
    \caption{Upper bound computation}
    \end{subfigure}%
    \begin{subfigure}{.45\textwidth}
      \resizebox{!}{5cm}{\begin{tikzpicture}
\draw [black!50!cyan,fill= white!90!cyan] (-2,0) .. controls (-0.8,-0.45) .. (-0.4,-0.7) -- (0,-2) .. controls (0.3,-1.6) .. (0,-0.7) -- (-1.75,3.6) .. controls (-1.85,3.8) .. (-2,4) -- (-1.8,3.4) .. controls (-2,2.85) .. (-4,2) -- cycle  node[anchor=south east] {Q};

\draw [black!50!green,fill = white!90!green] (0.25,-2) -- (-1,-2) -- (-1,4) -- (0.25,4) -- cycle node[anchor=south west] {Q'};

\draw [red, very thick, dashed] (-1,-3) -- (-1,5) node[anchor=north east] {Guard};

\draw [black!50!cyan, dashed] (-2,0) .. controls (-0.8,-0.45) .. (-0.4,-0.7) -- (0,-2) .. controls (0.3,-1.6) .. (0,-0.7) -- (-1.75,3.6) .. controls (-1.85,3.8) .. (-2,4) -- (-1.8,3.4) .. controls (-2,2.85) .. (-4,2) -- cycle ;

\draw [->] (-5,0) -- (3,0) node[anchor=north] {$x$} ;
\draw [->] (0,-3) -- (0,5) node[anchor=east] {$y$};
\foreach \x in {-4,-2,0,2} {
\draw (\x,-1pt) -- (\x,1pt)  node[anchor=north east] {\x};
}
\foreach \y in {-2,4} {
\draw (-1pt,\y) -- (1pt,\y) node[anchor=north east] {\y};
}
\end{tikzpicture}}
    \caption{Guard evaluation}
    \end{subfigure}
      
    \caption{Zonotopic concretization of operations on Quadratic Zonotopes}
    \label{fig:joinmeet_quad_zon}
  \end{figure}
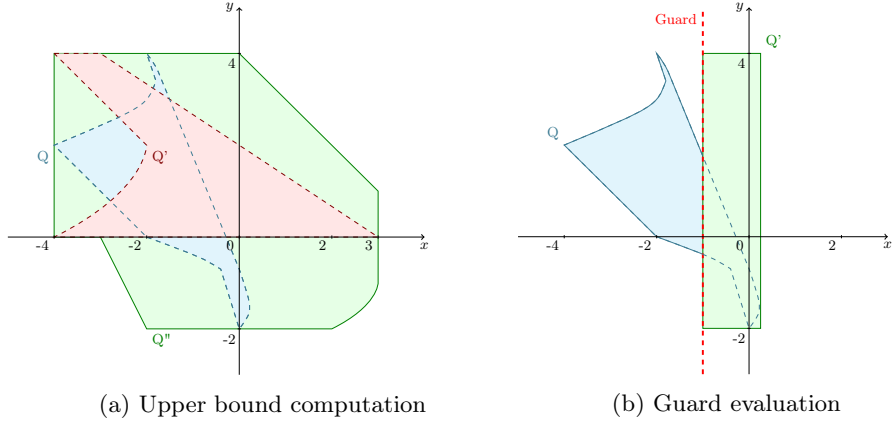

\subsubsection{Transfer functions.}

% Let us first introduce the grammar of expressions.

% \begin{definition}[Expressions]
%   \label{def:expr}
%   \[
%   \begin{array}{lcl}
%     Expr\;\;&::=&\;\; Cst\;\;|\;\;Var\;\;|\;\;Expr + Expr \;\;|\;\; - Expr \;\;|\;\; Expr \times Expr. \\
%     RelExpr\;\;&::=&\;\; Expr \;\; Comp \;\; 0 \\
%     Comp\;\;&::=&\;\;\leq\;\;|\;\;<\;\;|\;\;=\;\;|\;\;\geq\;\;|\;\;>
%   \end{array}
%   \]

% \end{definition}

%The two domain operators \texttt{guard} and \texttt{assign} over
The two operators \texttt{guard} and \texttt{assign} over
the expressions $RelExpr$ and $Expr$ are defined like in a non relational
abstract domain, as described
in~\cite[\textsection 2.4.4]{minephd}. Each operator relies on the forward
semantics of numerical expressions, computed within arithmetics operators in $\Q$:

\begin{definition}[Semantics of expressions]
Let $\mathcal{\Var}$ be a finite set of variables. We denote by $\llbracket\cdot \rrbracket_\Q (\Var \rightarrow \Q) \rightarrow
  \Q$ the semantics evaluation of an expression in an environment mapping
  variables to quadratic forms.
  \[
  \begin{array}{rcl}
    \llbracket v \rrbracket_\Q(Env) &=& \pi_k(Env) \textrm{ where } k \in [p] \text{ is
      the index of } v \in \Var \text{ in } Env\\
    \llbracket e_1 \; \operatorname{bop} \; e_2 \rrbracket_\Q(Env) &=& \llbracket e_1  \rrbracket_\Q(Env) \;\;
    \operatorname{bop}_\Q \;\; \llbracket e_2 \rrbracket_\Q(Env) \\ %\qquad op \in \{ +, \times \}\\
    %\llbracket - e \rrbracket_\Q(Env) &=& -_\Q \llbracket e  \rrbracket_\Q(Env)
    \llbracket \operatorname{uop} \; e \rrbracket_\Q(Env) &=& \operatorname{uop}_\Q \llbracket e  \rrbracket_\Q(Env)  
  \end{array}
  \]
\end{definition}

Guards, i.e. tests, are enforced through the classical combination of forward and
backward operators. Backward operators are the usual fallback operators,
e.g. $\llbracket x + y \rrbracket^\leftarrow = \left(x \sqcap_\Q (\llbracket x + y
\rrbracket -_\Q y), y \sqcap_\Q (\llbracket x + y
\rrbracket -_\Q x) \right)$ where $\sqcap_\Q$ denotes the meet of quadratic
forms. As for upper bound computation, no best lower bound exists and such meet
operator in $\Q$ has to compute a safe but
imprecise upper bound of maximal lower bounds. 

The meet over $\Q^m$ works as follows: it projects each argument to
intervals using $\gamma_\Q$, performs the meet computation and reinject with a fresh noise symbolthe resulting closed bounded interval to
$\Q$ thank to $\alpha_\Q$. Whereas, the meet over $\ZQ{p}{m}$ is the lift to the meet over $\Q^m$ to quadratic vectors. Formally:
\begin{definition}[$\sqcap_{\Q}, \sqcap_{\ZQ{}{}}$: Approximations of maximal lower bounds]
%We define the meet of quadratic forms as the following function: it projects each argument to
%intervals ($\gamma_\Q$), performs the meet computation ($\sqcap_{\INT}$) and reinject the resulting closed bounded interval to
%$\Q$ ($\alpha_\Q$): 
The meet $\sqcap_{\Q}: \Q^{m} \times \Q^{m} \rightarrow \Q^{1}$ is defined by:
\[
\forall x,y \in \Q^m, x \sqcap_\Q y \triangleq \alpha_\Q \left( \gamma_\Q (x)
\sqcap_{\INT} \gamma_\Q (y) \right).
\] 

%The meet over elements of $\ZQ{p}{m}$ is the lift of $\sqcap_\Q$ to vectors. We
%denote by $\sqcap_{\ZQ{}{}}$ such operator: 
The meet $\sqcap_{\ZQ{}{}}: \ZQ{p}{m} \times \ZQ{p}{m} \rightarrow \ZQ{p}{p}$ is defined,
for all $x,y \in \ZQ{p}{m}$ by $z= x \sqcap_{\ZQ{}{}} y \in \ZQ{p}{p}$ where:
\[
\forall i \in [p], z_i = \pi_i (x) \sqcap_{\Q} \pi_i (y) \textrm{ when } \pi_i
(x) \neq \pi_i (y), \pi_i (x) \textrm{ otherwise}.
\]
\end{definition}

%   each non equal
% component to its interval representation on which it performs the meet
% computation, before abstracting again to a fresh quadratic form.

% \begin{definition}[$\sqcap_{\ZQ{}{}}$: Approximation of maximal lower bounds]
%   Let $C$ and $C' \in \ZQ{p}{m}$. We define the meet
%   operator $\sqcap_{\ZQ{}{}}: \ZQ{p}{m} \times \ZQ{p}{m} \rightarrow
%   \ZQ{p}{m'}$ where $m' >= m$ such that $C'' = C \sqcap_{\ZQ{}{}} C'$ and
% \[\forall i \in [1,p], \pi_i(C''_i) = 
% \left\{
%   \begin{array}{l}
%     \pi_i (C) \textrm{ when } \pi_i (C) = \pi_i (C')\\
%     \alpha_\Q (\gamma_\Q(\pi_i (C)) \cap \gamma_\Q(\pi_i (C'))) \textrm { otherwise}
%   \end{array}
% \right.
% \]
% \end{definition}

\begin{example}
Let $Q$ be the following quadratic vector. The meet with the constraint $x + 1 \geq 0$ produces
the resulting quadratic vector $Q'$:
\[Q = \begin{cases}
  x = -1 + \epsilon_1 - \epsilon_2 - \epsilon_{1,1} \\
  y = 1 + 2 \epsilon_2 + \epsilon_{1,2}
\end{cases}
Q' = \begin{cases}
        x = - \frac{3}{8} + \frac{5}{8} \epsilon_3 \\
        y = 1 + 2 \epsilon_2 + \epsilon_{1,2}
      \end{cases}
\]

\end{example}

\begin{proof}
  $Guard(Q, x + 1 \geq 0) = Q \sqcap_{\ZQ{}{}} \left(  \alpha_\Q \left(
    \gamma_\Q (x +_\Q 1) \sqcap_{\INT} [0, +\infty [  \right) -_\Q 1 \right)$.
We use the more precise concretization over-approximation map $\gamma_\Q^{SDP}$ that will be
introduced in Sec.~\ref{sec-5}: $\gamma_\Q^{SDP} (\epsilon_1 - \epsilon_2 -
\epsilon_{1,1}) = [-3, 1.25]$. We focus on $x$ since the meet is performed
component-wise and $\alpha_\Q \left(
    \gamma_\Q (\epsilon_1 - \epsilon_2 - \epsilon_{1,1}) \sqcap_{\bINT} [0,
    +\infty [  \right) -_\Q 1 =   \alpha_\Q \left(
    [-3, 1.25] \sqcap_{\bINT} [0, +\infty [  \right) -_\Q 1 = 
  \alpha_\Q \left( [0, 1.25] \right) -_\Q 1 
=
  (5/8 + 5/8 \epsilon_3) -_\Q 1 
=  -3/8 + 5/8 \epsilon_3$ where $\epsilon_3$ is a fresh error term introduced
by $\alpha_\Q$.

\end{proof}

\section{Floating point computations}
\label{sec:floats}

All the operators presentation above assumed a real semantics. As usual when
analyzing programs, the domain has to be adapted to deal with floating point
arithmetics.

We recall that our use of quadratic zonotopes is to precisely over-approximate
reachable values as set of reals. We relied on the approach proposed by Stolfi
and De Figueiro~\cite{stolfi}, creating a new error term for each
operation. Other approaches such as generalized intervals~\cite{hansen} are
typically used in Fluctuat~\cite{DBLP:conf/sas/Goubault13}. Their definition in the quadratic setting is
given in~\cite{DBLP:journals/rc/MessineT06}. However, according to~\cite{stolfi}
the approach with error terms instead of interval arithmetics is more precise
but can generate an important number of error terms. 

In this specific case of quadratic forms, the term in $\epsilon_\pm$ is used to
accumulate floating point errors: the number of error terms does not evolve
due to floating point computation. The extension to zonotopes is direct since
numerical operations are evaluated at form level.

We illustrate the extension to quadratic form of~\cite{stolfi}:

\paragraph{Addition.}
According to Knuth~\cite[\S 4.2.2]{knuth2014art} algorithm, the exact computation of $u + v$ is $u +
v + e$ where $e = (u - ((u+v) - v)) + (v - ((u+v) - u))$ with all operations
performed in floating point arithmetics. Let $e^+(u,v)$ be such additive error $e$.

We consider the addition of two quadratic forms $x=(x_0, (x_i), (x_{ij}), x_\pm, x_+, x_-)$ and $y=(y_0, (y_i), (y_{ij}), y_\pm, y_+, y_-)$.
%where 
%\[
%x = x_0 + \displaystyle{\sum_{\substack{i,j=1}}^n} x_{ij} \epsilon_{ij} +
%\displaystyle{\sum_{i=1}^n} x_i \epsilon_i + x_\pm \epsilon_\pm + x_+ \epsilon_+ + x_- \epsilon_-
%\]
%and $\epsilon_\pm \in [-1;1], \epsilon_+ \in [0;1], \epsilon_- \in [-1, 0]$ and
%$x_\pm, x_+, x_- \geq 0$.
%%
%Similarly $y$ is characterized by $(y_0, (y_i), (y_{ij}), y_\pm, y_+, y_-)$.
%%
%% We denote the various additive error generated by the addition of $x$ and $y$ as $e^+_i
%% = e^+(x_i,y_i), e^+_{ij} = e^+(x_{ij},y_{ij}), e^+_+ = e^+(x_+,y_+), e^+_- =
%% e^+(x_-,y_-), e^+_\pm = e^+(x_\pm,y_\pm)$. 
%
The addition of $x$ and $y$ is modified to considered these generated errors:
\[
\begin{array}{l}
(x_0,(x_i), (x_{ij}),x_\pm, x_+, x_-) +_\Q (y_0, (y_i), (y_{ij}), y_\pm, y_+, y_-) =\\
    \hfill(x_0+y_0, (x_i+y_i), (x_{ij}+y_{ij}), x_\pm + y_\pm + rounded\_err, x_+ + y_+, x_- + y_-)\\
  \end{array}
\]
where
\begin{itemize}
\item $err =  \displaystyle{\sum_{\substack{i,j=1}}^n} e^+(x_{ij},y_{ij}) +
  \displaystyle{\sum_{i=0}^n} e^+(x_i,y_i) + e^+(x_\pm,y_\pm) + e^+(x_+,y_+) + e^+(x_-,y_-)$.
\item $rup$ denotes the rounding up;
\item $rounded\_err = max(|rup(err)|, |- rup(-err)|)$
\end{itemize}

\paragraph{External multiplication.}

Similarly, the algorithm of Dekker and Veltkamp characterizes the multiplicative error
obtained when computing $u \times v$. It relies on a constant $C$ depending on
the precision used. For single precision floats, $C = 2^{27} +1$. We denote  by
$e^\times(u,v)$ such multiplicative error and refer the interested reader to
Dekker's paper~\cite{dekker1971floating}.

The operator $*_\Q$ is modified to account such multiplicative errors:
\[
  \lambda *_\Q (x_0,(x_i), (x_{ij}),x_\pm, x_+, x_-) = (\lambda
    x_0,\lambda(x_i), \lambda(x_{ij}),|\lambda| x_\pm + r\_err , |\lambda| x_+, |\lambda| x_-)
\]
where \begin{itemize}
\item $err = \sum\limits_{i=1}^{n} e^\times(\lambda, x_i) +
  \sum\limits_{i,j=1}^{n} e^\times(\lambda, x_{ij}) +
  e^\times(\lambda, x_{\pm}) + e^\times(\lambda, x_-) + e^\times(\lambda, x_+)$.
\item $r\_err = max(|rup(err)|, |- rup(-err)|)$.

% \item $\alpha x = 
% (max(|\upharpoonright \sum\limits_{i=1}^{n} errtimes_i + \sum\limits_{i=1}^{n} \sum\limits_{j=1}^{n} errtimes_{ij} + errtimes_{\pm} + errtimes_- + errtimes_+|,|- \upharpoonright \sum\limits_{i=0}^{n} - errtimes_i + \sum\limits_{i=1}^{n} \sum\limits_{j=1}^{n} -errtimes_{ij} - errtimes_{\pm} - errtimes_- - errtimes_+|) \epsilon_{\pm}  $ (if $ \alpha \geq 0 $)

% \item $\alpha x = 
% (max(|\upharpoonright \sum\limits_{i=1}^{n} errtimes_i + \sum\limits_{i=1}^{n} \sum\limits_{j=1}^{n} errtimes_{ij} + errtimes_{\pm} + errtimes_- + errtimes_+|,|- \upharpoonright \sum\limits_{i=0}^{n} - errtimes_i + \sum\limits_{i=1}^{n} \sum\limits_{j=1}^{n} -errtimes_{ij} - errtimes_{\pm} - errtimes_- - errtimes_+|) \epsilon_{\pm} $ (if $ \alpha < 0 $)

\end{itemize}

All other operators behave similarly: each operation computing an addition or a
product generates an additive and a multiplicative error, respectively, accumulated
in the $x_\pm$ term.

%%% Local Variables: 
%%% mode: latex
%%% TeX-master: "submission"
%%% End: 

\section{Improving concretization using SDP}
\label{sec-5}
In this part, we propose a method based on semi-definite programming to compute an over-approximation of the interval concretization of a quadratic form. This method provides tighter bounds than $\mathbf{b}^q_{MT}$ and $\mathbf{B}^q_{MT}$ defined at Equation~\eqref{touamimessinebounds}.

Let consider a quadratic form $q=(c^q,(b^q)_m,(A^q)_m,c_{\pm}^q,c_+^q,c_-^q)\in \Q^m$.
Recall that $\ccm=[-1,1]^m\times [-1,1]\times [0,1]\times [-1,0]$, we remind that the concretization of $q$ is the interval defined $[\mathbf{b}^q,\mathbf{B}^q]$ where $\mathbf{b}^q=\inf\{q(x)\mid x\in \ccm\}$ and $\mathbf{B}^q= \sup\{q(x)\mid x\in \ccm\}$. 

In general, a standard quadratic form $r$ from $\rr^{m+3}$ to $\rr$ is defined by $x\mapsto r(x)=x^\intercal A^r x+{b^r}^\intercal x+c^r$
with a $(m+3)\times (m+3)$ symmetric matrix $A^r$, a vector of $\rr^{m+3}$, $b^r$ and a scalar $c^r$. 
We can cast $q$ into a standard quadratic form $r_q$, leading to $r_q(x)=q(x)$ for all $x\in\ccm$. Indeed, it suffices to take the following data :
\[
A^{r_q}=\left(\begin{array}{cc} \tilde{A} & 0_{m\times 3} \\ \multicolumn{2}{c}{0_{3\times (m+3)}} \end{array}\right) \text{ with }\tilde{A}=\dfrac{A^q+{A^q}^\intercal}{2},\ {b^{r_q}}^\intercal =\left({b^q}^\intercal,c_{\pm}^q,c_+^q,c_-^q\right) \text{ and } c^{r_q}=c^q
\] 
Let us denote by $\tr$, the trace function which associates to a matrix the sum of its diagonal elements and let $x\in\rr^{m+3}$.
A simple calculus yields to: 
\[
\label{eq_trace}
r_q(x)=\tr\left(\Mat^{r_q} X\right)\text{ where } \Mat^{r_q}=\begin{pmatrix} A^{r_q} & & \frac{1}{2}b^{r_q}\\ \frac{1}{2}{b^{r_q}}^\intercal & & c^{r_q}\end{pmatrix}\text{ and } X=\ux\uxt\enspace .
\]
%Semi-definite programming deals with matrices. 
To only deal with matrices, we have to translate the constraints on the vector $x$ into constraints on the matrix
$X$. Let us introduce the set $\mathcal{C}^m$ of $(m+4)\times (m+4)$ symmetric
matrices $Y$ such that:
%\begin{subequations}\label{X_constraint}
%\begin{gather}
%\forall\, i,j\in [m+3],\ i< j,\ Y_{i,j}\in\mathrm[-1,1\mathrm]\enspace,\label{eq_cont:1} \\
%\forall\, i\in [m+1],\ Y_{i,(m+4)}\in\mathrm[-1,1\mathrm]\enspace,\label{eq_cont:2}\\
%\forall\, i\in [m+3],\ Y_{i,i}\in\mathrm[0,1\mathrm]\enspace,\label{eq_cont:3}\\
%Y_{(m+2),(m+4)}\in\mathrm[0,1\mathrm]\enspace,\label{eq_cont:4}\\
%Y_{(m+3),(m+4)}\in\mathrm[-1,0\mathrm]\enspace,\label{eq_cont:5}\\  
%Y_{(m+4),(m+4)}=1\enspace .\label{eq_cont:6}
%\end{gather}
%\end{subequations}

\begin{center}
%\vspace{-1em}
%\begin{tabularx}{\textwidth}{X|X}
\begin{tabular}{p{.58\textwidth}|p{.4\textwidth}}
\begin{minipage}{0.55\textwidth}
%\centering
\begin{subequations}\label{X_constraint}
\vspace{-1em}
\begin{gather}
\forall\, i,j\in [m+3],\ i< j,\ Y_{i,j}\in\mathrm[-1,1\mathrm]\label{eq_cont:1} \\
\forall\, i\in [m+1],\ Y_{i,(m+4)}\in\mathrm[-1,1\mathrm]\label{eq_cont:2}\\
\forall\, i\in [m+3],\ Y_{i,i}\in\mathrm[0,1\mathrm]\label{eq_cont:3}
\end{gather}
\end{subequations}
\end{minipage}
&%\hfill
% \vline
% \hfill
\begin{minipage}{0.4\textwidth}
\addtocounter{equation}{-1}
\begin{subequations}\label{Y_constraint}
\addtocounter{equation}{3}
\vspace{-1em}
\begin{gather}
%\addtocounter{equation}{-3}
Y_{(m+2),(m+3)}\in\mathrm[-1,0\mathrm]\label{eq_cont:6}\\
Y_{(m+2),(m+4)}\in\mathrm[0,1\mathrm]\label{eq_cont:4}\\
Y_{(m+3),(m+4)}\in\mathrm[-1,0\mathrm]\label{eq_cont:5}\\  
Y_{(m+4),(m+4)}=1\label{eq_cont:6}
\end{gather}
\end{subequations}
\end{minipage}
%\end{tabularx}
\end{tabular}
\end{center}

Note by symmetry of $Y$, for all $i,j\in [m+3],\ i< j$, $Y_{j,i}\in [-1,1]$; for all $i\in [m+1],\ Y_{(m+4),i}\in [-1,1]$;
$Y_{(m+4),(m+3)}\in\mathrm[-1,0\mathrm]$ and $Y_{(m+4),(m+2)}\in\mathrm[0,1\mathrm]$.

\[
Y = [-1, 1] \cap % \left(\begin{array}{cccccc}
% %\hline
%   + \\
%   & \ddots && \top \\
%   & & \ddots \\
%   & \top & & + & - & +\\
% &&&- & + & -\\
% &&&+ & - & \{ 1 \}\\
% %\hline
% \end{array}\right)
%
\left(
\setlength\unitlength{13pt}
\begin{picture}(7,3)(0,-3.5)
\put(1,-1){\makebox(0,0){$+$}}
\put(2,-1.6){\makebox(0,0){$\ddots$}}
\put(3,-2.4){\makebox(0,0){$\ddots$}}
\put(1.8,-4.4){\makebox(0,0){$\top$}}
\put(5,-1.8){\makebox(0,0){$\top$}}
\put(4,-3.4){\makebox(0,0){$+$}}
\put(4,-4.4){\makebox(0,0){$-$}}
\put(4,-5.4){\makebox(0,0){$+$}}
\put(5,-3.4){\makebox(0,0){$-$}}
\put(5,-4.4){\makebox(0,0){$+$}}
\put(5,-5.4){\makebox(0,0){$-$}}
\put(6,-3.4){\makebox(0,0){$+$}}
\put(6,-4.4){\makebox(0,0){$-$}}
\put(6,-5.4){\makebox(0,0){$\{1\}$}}
\put(4,-6.5){\begin{rotate}{-45}\small$m+2$\end{rotate}}
\put(5,-6.5){\begin{rotate}{-45}\small$m+3$\end{rotate}}
\put(6,-6.5){\begin{rotate}{-45}\small$m+4$\end{rotate}}
\put(8,-3.4){\begin{rotate}{0}\small$m+2$\end{rotate}}
\put(8,-4.4){\begin{rotate}{0}\small$m+3$\end{rotate}}
\put(8,-5.4){\begin{rotate}{0}\small$m+4$\end{rotate}}
\put(4.5,-3){\line(0,-1){3}}
\put(5.5,-3){\line(0,-1){3}}
\put(3.5,-4){\line(1,0){3}}
\put(3.5,-4.9){\line(1,0){3}}
\end{picture}
\right)
\]
\vspace{1em}

We denote by $\mathbb{S}_{n}^+$ the set of semi-definite positive matrices of size $n\times n$ i.e. the $n\times n$ symmetric matrices $M$ such that for all $y\in\rr^{n}$, $y^\intercal M y\geq 0$. We recall that the rank of a matrix is the number of linearly independent rows (or columns). We denote by $\operatorname{rk}(M)$, the rank of the matrix $M$.

\begin{lemma}[Constraint translation]
\label{translation}
The following statement holds:
\[
\left\{X\in \mathbb{S}_{m+4}^+\mid\operatorname{rk}(X)=1,\ X\in \CCm\right\}=\left\{X\in\mathbb{S}_{m+4}^+\mid \exists\, x\in \ccm\ \st\ X=\ux\uxt\right\}\enspace .
\] 
\end{lemma}
\begin{proof}
%See Proof~\ref{translationproof}.
See Appendix.
\end{proof}
%\begin{proof}
%Let $X$ be a $(d+4)\times (d+4)$ matrix and assume there exists $x\in\rr^{d+3}$ such that $X=\ux \uxt$. Then, it is well-known that 
%$X\succeq 0$ and $\operatorname{rk}(X)=1$. Now by applying matrix calculus, we have: 
%\[
%X_{i,j}=\left\{
%\begin{array}{lr}
%x_i & \text{ if } i\in\{1,\ldots,d+3\} \land j=d+4\\
%x_j & \text{ if } j\in\{1,\ldots,d+3\} \land i=d+4\\
%x_i x_j & \text{ if } i,j\in\{1,\ldots,d+3\}\\
%1 & \text{ if } i=j=d+4
%\end{array}
%\right.
%\]
%We conclude that if $x\in C$ then $X\in\mathcal{C}$.

%Now, if we suppose that $X$ is a $(d+4)\times (d+4)$ matrix such that $\operatorname{rk}(X)=1$ and $X\succeq 0$ and $X\in\mathcal{C}$ then there exists $z\in C$ such that $X=\begin{pmatrix} x\\ 1 \end{pmatrix} \begin{pmatrix} x\\ 1\end{pmatrix}^\intercal$.
%\qed
%\end{proof}
Lemma~\ref{translation} allows to conclude that optimizing $r_q$ over $\ccm$ and optimizing $X\mapsto \tr(\Mat^{r_q} X)$ over 
$\{X\in \mathbb{S}_{m+4}^+\mid \operatorname{rk}(x)=1,\ X\in \CCm \}$ is the same. However, the rank one constraint on $X$ leads to a non-convex problem which makes it difficult to solve. A natural and a commonly used relaxation is to remove the rank constraint to get a linear problem over semi-definite positive matrices. This discussion is formulated as Proposition~\ref{coro}.
\begin{proposition}
\label{coro}
The interval bounds of the concretization of $q$ can be computed from the two following non-convex semi-definite programs:
\[
\begin{array}{llccr}
\mathbf{b}^q&=&\inf & & tr(\Mat^{r_q} X)\\
& & & \st& \left\{
\begin{array}{r}
X\in\CCm\\ 
X\in \mathbb{S}_{m+4}^+\\
\operatorname{rk}(X)=1
\end{array}
\right.
\end{array}
\qquad \text{ and } \qquad  
\begin{array}{llccr}
\mathbf{B}^q&=&\sup&  & tr(\Mat^{r_q} X)\\
& & & \st&\left\{
\begin{array}{r}
 X\in\CCm\\
 X\in\mathbb{S}_{m+4}^+\\
\operatorname{rk}(X)=1
\end{array}
\right.
\end{array}
\] 
%\[
%\mathbf{b}^q=\inf\{tr\left(\Mat^{r_q} X\right)\mid X\in \CCm,\ X\in\mathbb{S}_{m+4}^+,\ \operatorname{rk}(X)=1\}
%\]
%and
%\[
%\mathbf{B}^q=\sup\{tr\left(\Mat^{r_q} X\right)\mid X\in \CCm,\ X\in\mathbb{S}_{m+4}^+,\ \operatorname{rk}(X)=1\}\enspace.
%\]
By removing the rank constraint:
%~ \[
%~ \mathbf{b}^q_{SDP}:=\inf\{tr\left(\Mat^{r_q} X\right)\mid X\in \CCm,\ X\in\mathbb{S}_{m+4}^+\}\leq \mathbf{b}^q
%~ \]
%~ and
%~ \[
%~ \mathbf{B}^q_{SDP}:=\sup\{tr\left(\Mat^{r_q} X\right)\mid X\in \mathcal{C}_m,\ X\in\mathbb{S}_{m+4}^+\}\geq \mathbf{B}^q\enspace.
%~ \]
\[
\begin{array}{llccrrr}
\mathbf{b}^q_{SDP}&=&\inf & & tr(\Mat^{r_q} X)&\leq& \mathbf{b}^q\\
& & & \st& \left\{
\begin{array}{r}
X\in\CCm\\ 
X\in \mathbb{S}_{m+4}^+
\end{array}
\right.
\end{array}
\text{ and }  
\begin{array}{llccrrr}
\mathbf{B}^q_{SDP}&=&\sup&  & tr(\Mat^{r_q} X)&\geq& \mathbf{B}^q\\
& & & \st&\left\{
\begin{array}{r}
 X\in\CCm\\
 X\in\mathbb{S}_{m+4}^+
\end{array}
\right.
\end{array}
\]
\end{proposition}
Finally, the interval bounds of the concretization are safely approximated by using $\mathbf{b}^q_{SDP}$ and $\mathbf{B}^q_{SDP}$
and we write $\gamma_\Q^{SDP}(q)\triangleq [\mathbf{b}^q_{SDP},\mathbf{B}^q_{SDP}]$. Moreover, those bounds improve the ones provides by~\cite{DBLP:journals/rc/MessineT06}. 
\begin{theorem}[Bounds improvements]
\label{improvboundsth}
Let $q\in\Q^m$. The following inequalities hold:
\[
\gamma_\Q(q)\subseteq\gamma_\Q^{SDP}(q)\subseteq \gamma_\Q^{MT}(q)\text{ i.e. }\ \mathbf{b}^q_{MT}\leq \mathbf{b}^q_{SDP}\leq \mathbf{b}^q\ \land\ \mathbf{B}^q\leq \mathbf{B}^q_{SDP}\leq \mathbf{B}^q_{MT}\enspace .
\]
\end{theorem}
\begin{proof}
%See Proof~\ref{proofboundsimprov}.
See Appendix.
\end{proof}
In term of complexity, SDP problems can be solved in polynomial time to an arbitrary prescribed
precision by the ellipsoid method~\cite{grotschel1988geometric}. More precisely, let $\alpha> 0$ be a given rational,
suppose that the input data of a semi-definite program are rational and suppose that an
integer $N$ is known, such that the feasible set lies inside the ball of the radius $N$ around
zero. Then a feasible solution -- the value of which is at most
at a distance $\alpha$ from the optimal value -- can be found in a time that is polynomial in the
number of bits of the input data and in $-log(\alpha)$. This latter feasible solution can be found in polynomial time by interior point methods~\cite{nesterov1994interior} if a strictly feasible solution is available. The advantage of interior methods is that they are very efficient in
practice. We refer the reader to~\cite{ramana1996semidefinite} for more information.
%However, when the input is rational, no size on the bit lengths of the
%intermediate data is currently known, so that the term \emph{polynomial time} for interior point
%methods is only understood in the model of computation over real numbers (rather than
%in bit mode~\cite{DBLP:books/fm/GareyJ79}). 
\begin{corollary}
The reals $\mathbf{b}^q_{SDP}$ and $\mathbf{B}^q_{SDP}$ can be computed in polynomial time.
\end{corollary}
The Figure~\ref{fig:projection} illustrates such concretization on the quadratic
zonotopes defined in Example~\ref{ex:quadzon}.
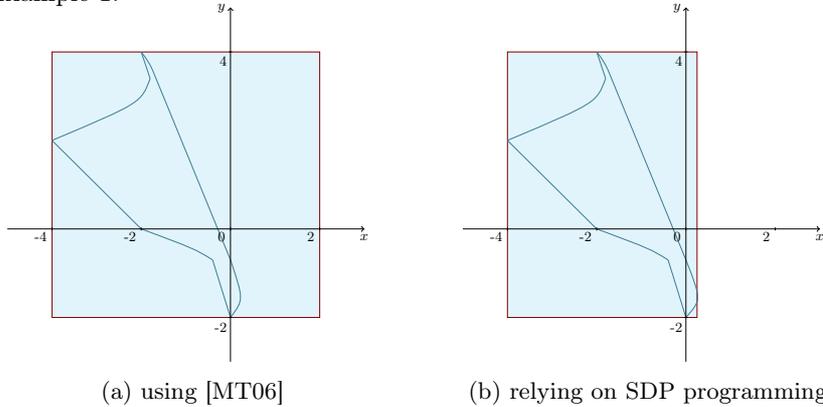
\begin{figure}
\vspace{-1cm}
  \centering
\begin{subfigure}{0.4\textwidth}
  \resizebox{\textwidth}{!}{\begin{tikzpicture}
\draw [black!50!red,fill = white!90!cyan] (-4,-2) -- (2,-2) -- (2,4) -- (-4,4) -- cycle;

\draw [black!50!cyan] (-4,2) -- (-2,0);
\draw [black!50!cyan] (-2,0) .. controls (-0.8,-0.45) .. (-0.4,-0.7);
\draw [black!50!cyan] (-0.4,-0.7) -- (0,-2);
\draw [black!50!cyan] (0,-2) .. controls (0.3,-1.6) .. (0,-0.7);
\draw [black!50!cyan] (0,-0.7) -- (-1.75,3.6);
\draw [black!50!cyan] (-1.75,3.6) .. controls (-1.85,3.8) .. (-2,4);
\draw [black!50!cyan] (-2,4) -- (-1.8,3.4);
\draw [black!50!cyan] (-1.8,3.4) .. controls (-2,2.85) .. (-4,2);

\draw [->] (-5,0) -- (3,0) node[anchor=north] {$x$} ;
\draw [->] (0,-3) -- (0,5) node[anchor=east] {$y$};
\foreach \x in {-4,-2,0,2} {
\draw (\x,-1pt) -- (\x,1pt)  node[anchor=north east] {\x};
}
\foreach \y in {-2,4} {
\draw (-1pt,\y) -- (1pt,\y) node[anchor=north east] {\y};
}
\end{tikzpicture}}
  \caption{using~\cite{DBLP:journals/rc/MessineT06}}
  \end{subfigure}\hspace{1cm}
\begin{subfigure}{0.4\textwidth}
  \resizebox{\textwidth}{!}{\begin{tikzpicture}
\draw [black!50!red,fill = white!90!cyan] (-4,-2) -- (0.25,-2) -- (0.25,4) -- (-4,4) -- cycle;

\draw [black!50!cyan] (-4,2) -- (-2,0);
\draw [black!50!cyan] (-2,0) .. controls (-0.8,-0.45) .. (-0.4,-0.7);
\draw [black!50!cyan] (-0.4,-0.7) -- (0,-2);
\draw [black!50!cyan] (0,-2) .. controls (0.35,-1.6) .. (0,-0.7);
\draw [black!50!cyan] (0,-0.7) -- (-1.75,3.6);
\draw [black!50!cyan] (-1.75,3.6) .. controls (-1.85,3.8) .. (-2,4);
\draw [black!50!cyan] (-2,4) -- (-1.8,3.4);
\draw [black!50!cyan] (-1.8,3.4) .. controls (-2,2.85) .. (-4,2);

\draw [->] (-5,0) -- (3,0) node[anchor=north] {$x$} ;
\draw [->] (0,-3) -- (0,5) node[anchor=east] {$y$};
\foreach \x in {-4,-2,0,2} {
\draw (\x,-1pt) -- (\x,1pt)  node[anchor=north east] {\x};
}
\foreach \y in {-2,4} {
\draw (-1pt,\y) -- (1pt,\y) node[anchor=north east] {\y};
}
\end{tikzpicture}}
  \caption{relying on SDP programming}
  \end{subfigure}
  \caption{Projection to intervals}
\label{fig:projection}
\end{figure}

In terms of related works, the use of semidefinite programming to compute interval concretisation
of
nonlinear operation for affine forms already appeared
in~\cite[Prop. 5.1.2]{ghorbal:pastel-00643442}. 
This approach appears to be the dual version of the semidefinite
programs that we presented in this paper.

%%% Local Variables: 
%%% mode: latex
%%% TeX-master: "submission"
%%% End: 

\vspace{-1cm}
\section{Experimentation}
\label{sec:exp}

All presented materials has been implemented in an open-source tool written in
OCaml% and are available at \url{https://cavale.enseeiht.fr/QuadZonotopes/}.
\footnote{Tool and experiments available at
  \url{https://cavale.enseeiht.fr/QuadZonotopes/} }.  This tool is used for
teaching purpose and only consider simple imperative programs without function
calls. It implements interval analysis, affine and quadratic zonotopes and
provides a binding to APRON to evaluate the more complex T1P
domain~\cite{DBLP:conf/cav/GhorbalGP09}, an affine zonotope domain with
constraints. The reduced concretization is not integrated through the use of the
CSDP or Mosek SDP solvers. Due to the increase cost in terms of computation
time, it can be globally or locally activated when calling the analyzer.

The quadratic zonotope domain has been evaluated on examples provided in APRON
T1P source code, or Fluctuat distribution, as well as simple iterative
schemes. We present here the results obtained on an arctan function, the example
of~\cite{Comba93affinearithmetic} and the
Householder function analyzed in~\cite{DBLP:conf/cav/GhorbalGP09}.

Let us first consider the arctan function defined in Figure~\ref{ex:arctan} and
the analysis results in Table~\ref{tab:arctan_res}.

\begin{figure}
\centering
\begin{subfigure}{.65\textwidth}
\begin{lstlisting}[basicstyle=\small]
if (x > 1.) {	    	
 y = 1.5708 - 1/x*(1-C1/$x^2$+C2/$x^4$+C3/$x^6$+
	C4/$x^8$+C5/$x^{10}$+C6/$x^{12}$+C7/$x^{14}$+C8/$x^{16}$)
}
if (x < 1.) {	    	
 y = -1.5708 - 1/x*(1-C1/$x^2$+C2/$x^4$+C3/$x^6$+
	C4/$x^8$+C5/$x^{10}$+C6/$x^{12}$+C7/$x^{14}$+C8/$x^{16}$)
}
else {	    	
 y = x*(1-C1*$x^2$+C2*$x^4$+C3*$x^6$+
	C4*$x^8$+C5*$x^{10}$+C6*$x^{12}$+C7*$x^{14}$+C8*$x^{16}$)
}
\end{lstlisting}
\end{subfigure}
\begin{subfigure}{0.3\textwidth}
{\small
with the constants defined as:
}
\medskip

\centering
{\small
\begin{tabular}{ l c r } 
  C1 &~~& $0.0028662257$ \\      
  C2 && $-0.0161657367$ \\     
  C3 && $0.0429096138$ \\  
  C4 && $-0.0752896400$ \\ 
  C5 && $0.1065626393$ \\ 
  C6 && $-0.1420889944$ \\ 
  C7 && $0.1999355085$ \\ 
  C8 && $-0.3333314528$ \\ 
\end{tabular}
}
\end{subfigure}
\vspace{-2em}
\caption{Arctan program}
\label{ex:arctan}
\end{figure}

\begin{table}
%\vspace{-1cm}
{\small  \centering
%   \begin{tabular}{|l||c|C{.8cm}||c|C{.8cm}|}
% \cline{2-5} 
%     \multicolumn{1}{c}{}& \multicolumn{2}{|c||}{$x \in [-1, 1]$ } & \multicolumn{2}{|c|}{$x \in [-10, 10]$ } \\ \hline
%     Domain & Bounds & Time (ms) & Bounds & Time (ms)\\ \hline
%     Interval & [-1.919150, 1.919150] & 15 & [-1.919150, 1.919150] & 17 \\ \hline     
%     Affine Zonotopes & [-1.9191498, 1.9191498] & 26 & [-2.3648468, 2.3648468] &
%     26 \\ \hline
%     Quadratic Zonotopes &  [-1.0028662, 1.0028662] & 38 & [-1.5975015, 1.5917690] & 39 \\ \hline 
%     Apron T1P & [-1.3494078, 1.3494078] & 96 & [-1.4775358, 1.4775358]  & 104 \\ \hline
% \end{tabular}
\begin{tabularx}{\textwidth}{|X|c||c|}
\cline{2-3}
\multicolumn{1}{c|}{}& $x \in [-1, 1]$  & $x \in [-10, 10]$  \\ \hline
Domain & Bounds & Bounds \\ \hline
Interval &
[-1.919149, 1.919149] &
[-1.919149, 1.919149]
\\ \hline
Affine Zonotopes &
[-1.919149, 1.919149] &
[-2.364846, 2.364846]
\\ \hline
Quadratic Zonotopes &
[-1.002866, 1.002866] &
[-1.597501, 1.591769]
\\ \hline
Affine Constrainted Zonotopes (Apron) &
[-1.349407, 1.349407] &
[-1.477535, 1.477535]
\\ \hline
\end{tabularx}

 }
% \begin{result} case x = rand(-10,10) \\
% \begin{tabular}{ l | c | c } 
% Domain & Bounds & Time cost (ms)\\ \hline
% \hline                     
%   Interval domain & [-1.919150, 1.919150] & 17 \\ \hline     
%   Affine set domain & [-2.3648468, 2.3648468] & 26 \\ \hline    
%   Quadratic set domain & [-1.5975015, 1.5917690] & 39 \\ \hline 
%   Apron T1P & [-1.4775358, 1.4775358]  & 104 \\ \hline

  \caption{Arctan program analysis results}
  \label{tab:arctan_res}
\end{table}

In~\cite{Comba93affinearithmetic} the authors considered the function $\sqrt(x^2
+ x - 1/2) / \sqrt(x^2 + 1/2)$ and the precision obtained using affine
arithmetics while evaluating the function on a partition of the input range as
sub-intervals. Figure~\ref{fig:stolfi} illustrates the obtained results with our different
abstractions for a division from $1$ to $14$, higher partition divisions
(eg. 500) converge in terms of precision. Quadratic zonotopes shows here to be a good
alternative to interval or affine zonotopes abstractions.

%\vspace{-1cm}
About the Householder function, it converges towards $1/ \sqrt{A}$:
\[
\begin{array}{l}
x_0 = 2^{-4}\\
x_{n+1} = x_n (1+\frac{1}{2}(1 - A x_n^2) + \frac{3}{8}(1 - A x_n^2)^2 )
\end{array}
\]
We analyzed it using loop unrolling with $A \in [16, 20]$ and compared the
global errors obtained at the $i-th$ iterate: the difference between the max and
min values. Figure~\ref{fig:househ} presents the precision obtained with
different analyses. Quadratic zonotopes provides here better bounds than affine
or interval analysis and shows to scale better than all other analyses.

For the sake of comparison we also compared with the zonotope analysis provided
in Apron as the T1P (Taylor 1 plus) abstract
domain~\cite{DBLP:conf/cav/GhorbalGP09}. We recall that this domain is not just
based on affine arithmetics but also embed linear relationships between error
terms in the zonotope. It shows better performance than affine zonotopes and a
similar extension of our quadratic zonotopes should be considered.

Finally a disappointing result is the cost of the optimized concretization. This
algorithm can potentially be used when converting a quadratic form to intervals
and impact the computation of meet and join operators. However since meet
operators are widely used in backward semantics~\cite{minephd}, the global cost
impacts widely the timing results. Moreover, in most cases the precision is not
widely impacted.

\begin{figure}
     \centering
     \begin{subfigure}[t]{.48\textwidth}
       \resizebox{\textwidth}{!}{\begin{tikzpicture}
    \begin{semilogyaxis}[
        xlabel=\# of subdivisions,
        ylabel=Global Error
    ]
\addplot plot coordinates {
(1, inf)
(2, 3.60555127546)
(3, 5.38515366544)
(4, 2.2360679775)
(5, 2.18118677732)
(6, 2.18817574984)
(7, 1.89644276761)
(8, 1.82574185835)
(9, 1.73582044704)
(10, 1.71446607998)
(11, 1.6742485743)
(13, 1.61553612293)
};
\addplot plot coordinates {
(3, 13.7982240345)
(5, 3.467800497)
(6, 3.63019491776)
(7, 2.04587966863)
(8, 1.79059409586)
(9, 1.62973296127)
(10, 1.50717293419)
(11, 1.40619933209)
(13, 1.28900163493)
};
\addplot plot coordinates {
(1, 4.10493351595)
(3, 3.87478462516)
(4, 3.5979770199)
(5, 2.2982996548)
(6, 2.19160508727)
(7, 1.68129007725)
(8, 1.58957380192)
(9, 1.41967364462)
(10, 1.39280645221)
(11, 1.28189308341)
(13, 1.20505716694)
};
\addplot plot coordinates {
(1, 3.60555127547)
(2, 2.511913348)
(3, 1.75167279806)
(4, 1.36157227004)
(5, 1.20562362195)
(6, 1.27066983516)
(7, 1.01142220257)
(8, 0.929690695564)
(9, 0.932346557942)
(10, 0.888854328833)
(11, 0.854197731019)
(13, 0.839039943028)
};
    \legend{Intervals, Zonotopes, Quad. Zonotopes, Apron T1P}
    \end{semilogyaxis}
\end{tikzpicture}}
       \caption{Stolfi~\cite{Comba93affinearithmetic} example evaluated on partitioned input range}
       \label{fig:stolfi}
     \end{subfigure}\hfill
     \begin{subfigure}[t]{.48\textwidth}
       \resizebox{\textwidth}{!}{\begin{tikzpicture}
    \begin{semilogyaxis}[
        xlabel=\# Iterations,
        ylabel=Global Error
    ]
\addplot plot coordinates {
(1, 0.00116920471191)
(2, 0.00849903242146)
(3, 0.0378319517089)
(4, 0.11731740504)
(5, 0.390339893088)
(6, 2.57667754578)
(7, 3017.31624606)
(8, 5.89192447161e+18)
(9, 1.67278001426e+95)
(10, inf)
(11, inf)
(12, inf)
(13, inf)
(14, inf)
(15, inf)
(16, inf)
(17, inf)
};
\addplot plot coordinates {
(1, 0.00117206573486)
(2, 0.00752962073792)
(3, 0.0240712391138)
(4, 0.0411152128624)
(5, 0.058064974194)
(6, 0.0803133900482)
(7, 0.120042280664)
(8, 0.225909649656)
(9, 0.808751533444)
(10, 28.6407257213)
(11, 196007551.195)
(12, 2.7122848369e+42)
(13, 1.37609290784e+213)
};
\addplot plot coordinates {
(1, 0.00117063522339)
(2, 0.00747463330942)
(3, 0.0225617484874)
(4, 0.0297545813698)
(5, 0.0328425539363)
(6, 0.0362124859843)
(7, 0.0401931385455)
(8, 0.0451515973008)
(9, 0.051822058007)
(10, 0.0617779091804)
(11, 0.0786636867022)
(12, 0.112059525806)
(13, 0.193654736744)
(14, 0.488953796065)
(15, 3.69142735195)
(16, 5152.30326609)
(17, 5.71116344207e+18)
};
\addplot plot coordinates {
(1, 0.00116920471192)
(2, 0.00749781996004)
(3, 0.0236599551905)
(4, 0.0386003808441)
(5, 0.051518562757)
(6, 0.0659521985605)
(7, 0.0855290581799)
(8, 0.118476297962)
(9, 0.192855982409)
(10, 0.469935974007)
(11, 3.69667506509)
(12, 13562.14175)
(13, 9.31816173732e+21)
(14, 1.42672434094e+111)
(15, inf)
(16, inf)
(17, inf)
};
    \legend{Intervals, Zonotopes, Quad. Zonotopes, Apron T1P}
    \end{semilogyaxis}
\end{tikzpicture}}
       \caption{Householder precision wrt. number of unrolling.}
       \label{fig:househ}
     \end{subfigure}%
%\end{figure}
% \begin{figure}
%      \centering
%      \begin{subfigure}[b]{.42\textwidth}
%        \resizebox{\textwidth}{!}{\input{bench_plot_precision.tikz}}
% %       \includegraphics[width=\textwidth]{PrecisionHHwithout.png}
%      \end{subfigure}%
% %
%      \begin{subfigure}[b]{.42\textwidth}
%        \resizebox{\textwidth}{!}{\input{bench_plot_time.tikz}}
%        %\includegraphics[width=\textwidth]{TimeCostHHwithout.png}
%      \end{subfigure}
% %\vspace{-1em}
%     \caption{Householder analysis results}
\end{figure}

% \begin{result} Program 3, case A = rand(16,20) \\

% \begin{figure}[h]
% 	\centering
%     \includegraphics[
%   width=10cm,height=6cm]{PrecisionHHwith.png}
% \end{figure}
% \end{result}

%\vspace{-1em}
Generally speaking, T1P abstract domain gives in most of the case better bounds
in presence of several guards in the program. It provides a way to keep the
noise symbols when computing "meets" in such a way that the relationships between the
variables are maintained. On the other hand, the quadratic zonotope abstract
domain is very efficient when working with polynomial program computing
non-linear operations such as products. It is also more effective than the
affine set abstract domain in term of precision but our implementation does not
rely on the constrained mechanisms developed in T1P.

All experiments are developed in the Appendix~\ref{sec:exp_appendix}.

%%% Local Variables: 
%%% mode: latex
%%% TeX-master: "submission"
%%% End: 

%  LocalWords:  zonotope affine arithmetics zonotopes concretization Stolfi wrt

\section{Conclusion}
\label{sec-7}

Zonotopic abstractions are the current more promising analyses when it comes to
the formal verification of floating point computations such as the ones found in aircraft
controllers. The presented analysis seems an interesting alternative to affine zonotopes,
increasing precision while keeping the complexity quadratic in the
number of error terms. Quadratic zonotopes seems more suited than linear
abstractions when analyzing non linear functions such as
multiplications. Among the zoology of abstract domains, they belong to the small
set of algebraic domains with non convex and non symmetric concretization. This may be later of
great impact, e.g. when considering properties involving positivity of products of
negative error terms.

\vspace{-0,1cm}

\paragraph{Perspectives.}
On the theoretical side, it would be interesting to compare the abstraction
generated by quadratic form with respect to the classical zonotopes, generated
by affine forms. While graphically speaking quadratic zonotopes seem strictly
included in their affine counterpart, the existence of a Galois connection
between the two abstraction is non trivial to exhibit, if ever it exists.

On the application side, our comparison in the benchmarks with affine zonotopes
was a little bit biased (against us) since we considered a naive meet operator like
the one we provided. The work of~\cite{DBLP:conf/cav/GhorbalGP09},
represented as Apron T1P in the benchmark evaluation, proposes to enrich
zonotopes with linear constraints over error terms to encode intersection. This
extension of the domain seems a feasible approach in our setting. It would then
allow a stronger comparison between affine and quadratic zonotopes.

Last, both affine and quadratic arithmetics can be seen, respectively, as a
first and second order Taylor polynomial abstraction. It would be interesting to
evaluate how this approach can be extended and how it combines with other
methods aiming at regaining precision such as branch-and-bound algorithms.

%%% Local Variables: 
%%% mode: latex
%%% TeX-master: "submission"
%%% End: 

\bibliographystyle{alpha}
\bibliography{biblio}

\newpage

\appendix
\section*{Appendix}
%\section{Appendix}
\subsection*{Proofs of Join Soundness}

\begin{proof}[Proof of Property~\ref{extensionprop}]
\label{proofextprop}
\begin{enumerate}
\item Let $i,j\in\mathbb{N}$. Let $\epsilon\in [-1,1]^m$ and $\epsilon'\in [-1,1]^{m+i+j}$ such that for all $i+1\leq k \leq m+i$$\epsilon'_k=\epsilon_{k-i}$. Let $(\epsilon_{\pm},\epsilon_+,\epsilon_-)\in [-1,1]\times [0,1]\times [-1,0]$. 
%Then 
%$\ext_{i,j}(q)(\epsilon',\epsilon_{\pm},\epsilon_+,\epsilon_-)=\sum_{l=1}^{m+i+j} \sum_{n=1}^{m+i+j} \epsilon'_l A'_{l,n} \epsilon'_n
%+\sum_{n=1}^{m+i+j} b'_{n} \epsilon'_n+c_{\pm} \epsilon_{\pm}+c_+\epsilon_++c_-\epsilon_-$. 
By definition of $A'$, $b'$ and $\epsilon'$, we have:

$
\begin{array}{ll}
&\ext_{i,j}(q)(\epsilon',\epsilon_{\pm},\epsilon_+,\epsilon_-)\\
=&\displaystyle{\sum_{l=i+1}^{m+i} \sum_{n=i+1}^{m+i} \epsilon'_l A'_{l,n} \epsilon'_n+\sum_{n=i+1}^{m+i} b'_{n} \epsilon'_n+c_{\pm} \epsilon_{\pm}+c_+\epsilon_++c_-\epsilon_-}\\
=&\displaystyle{\sum_{l=i+1}^{m+i} \sum_{n=i+1}^{m+i} \epsilon_{l-i} A_{l-i,n-i} \epsilon_{n-i}+\sum_{n=i+1}^{m+i} b_{n-i} \epsilon_{n-i}+c_{\pm} \epsilon_{\pm}+c_+\epsilon_++c_-\epsilon_-}\\
=&\displaystyle{\sum_{l=1}^{m} \sum_{n=1}^{m} \epsilon_l A_{l,n} \epsilon_n+\sum_{n=1}^{m} b_{n} \epsilon_n+c_{\pm} \epsilon_{\pm}+c_+\epsilon_++c_-\epsilon_-}=q(\epsilon,\epsilon_{\pm},\epsilon_+,\epsilon_-)
\end{array}
$.
\item From the first point: 
\begin{itemize}
\item $\forall\, t\in\ccm$, $\exists\, t'\in\cc^{m+p} \st q(t)=\ext_{i,j}(q)(t')$ hence $\gamma_\Q(q)\subseteq\gamma_\Q(\ext_{i,j}(q))$.
\item $\forall\, t'\in\cc^{m+p}$, $\exists\, t\in\ccm \st q(t)=\ext_{i,j}(q)(t')$ hence $\gamma_\Q(q)\supseteq \gamma_\Q(\ext_{i,j}(q))$.
\end{itemize}
\end{enumerate}
%\begin{enumerate}
% \item $\forall\, t\in\ccm$, $\exists\, t'\in\cc^{m+p}\ \st\ q(t)=\ext_{i,j}(q)(t')$ and thus $\gamma_\Q(q)\subseteq\gamma_\Q(\ext_{i,j}(q))$.
 %\item $\forall\, t'\in\cc^{m+p}$, $\exists\, t\in\ccm\ \st\ q(t)=\ext_{i,j}(q)(t')$ and thus $\gamma_\Q(q)\supseteq \gamma_\Q(\ext_{i,j}(q))$.
%\end{enumerate}

\end{proof}

%\begin{lemma}
%\label{prooflemmajoin}
%By construction: 
%\[
%\gamma_{\ZQ{}{}}(\ext_{m+p}(q'')+q^{err})=\gamma_{\ZQ{}{}}(q'')\oplus \prod_{k=1}^p \gamma_{\Q^{m+p}}(q_k^{err})
%\]
%where $A\oplus B=\{a+b\mid a\in A,\ b\in B\}$ and $\prod_i^n A_i=\{(a_1,\ldots,a_n)\mid a_i\in A_i\}$. 
%\end{lemma}
\begin{proof}[Proof of Lemma~\ref{lemmajoin}]
\label{prooflemmajoin}
By definition, $\alpha_\Q$ creates a fresh noise symbol when it is called, then $q_k^{e}$ do not share noise symbols with anyone else and a $q_k^{e}$ only depends on the $m+k$-th noise symbol:
\begin{equation}
\label{eqqerr}
\begin{array}{l}
\forall\, (\epsilon_{\pm},\epsilon_+,\epsilon_-), (\epsilon'_{\pm},\epsilon'_+,\epsilon'_-)\in [-1,1]\times [0,1]\times [-1,0],\ \forall\, k\in [p],\\ 
\forall\, u,v\in [-1,1]^{m+p}\ \st\ u_{m+k}=v_{m+k},\ q_k^{e}(u,\epsilon_{\pm},\epsilon_+,\epsilon_-)= q_k^{e}(v,\epsilon'_{\pm},\epsilon'_+,\epsilon'_-)
\end{array}
\end{equation}
Let $t=(\epsilon,\epsilon_{\pm},\epsilon_+,\epsilon_-)\in\ccmp$. Define $\epsilon''\in [-1,1]^{m}$ such that $\epsilon''_k=\epsilon_k$ for all $k\in [m]$ and $t''=(\epsilon'',\epsilon_{\pm},\epsilon_+,\epsilon_-)\in\ccm$. From Property~\ref{extensionprop} and since $q''\in\ZQ{p}{m}$, for all $k\in[p]$, $q''_k(t'')=\ext_{0,p}(q''_k)(t)$.
Now define for all $k\in [p]$, $\epsilon^k\in[-1,1]^{m+p}$ such that $\epsilon_i^k=\epsilon_{m+k}$ if $i=m+k$ and 0 otherwise and $t^k=(\epsilon^k,\epsilon'_{\pm},\epsilon'_+,\epsilon'_-)$. From Equation~\eqref{eqqerr}, for all $k\in\{1,\ldots,p\}$, $q_k^{e}(t)=q_k^{e}(t^k)$. Hence, for all $k\in [p]$, $(\ext_{0,p}(q''_k)+q_k^e)(t)=\ext_{0,p}(q''_k)(t)+q_k^e(t)=q''_k(t'')+q_k^e(t^k)$. 
Finally, for all $t\in\ccmp$, there exists $t''\in\ccmp$ and $t^1,\ldots,t^p\in\ccmp$ such that for all $k\in[p]$, $\ext_{0,p}(q_k'')(t)+q_k^{e}(t)=q''_k(t'')+q_k^e(t^k)$ and thus $\gamma_{\ZQ{}{}}((\ext_{0,p}(q''_k))_{k\in [p]}+q^{e})\subseteq \gamma_{\ZQ{}{}}(q'')\oplus \prod_{k=1}^p \gamma_{\Q^{m+p}}(q_k^{err})$.  

Now let us take $t''=(\epsilon'',\epsilon_{\pm},\epsilon_+,\epsilon_-)\in\ccm$ and $t^1,\ldots,t^p\in\ccmp$. Let us define $t=(\epsilon,\epsilon_{\pm},\epsilon_+,\epsilon_-)\in\ccmp$ by for all $i\in [m+p]$, $\epsilon_i=\epsilon''_i$
if $i\in [m]$ and $\epsilon_i=\epsilon^i$ if $m+1\leq i\leq m+p$, then from Property~\ref{extensionprop}, $\ext_{0,p}(q'')(t)=q''(t'')$ and from Equation~\eqref{eqqerr}, $q^{e}(t)=(q_k^e(t^k))_{1\leq k\leq p}$. Finally, for all 
$t''\in\ccm$ and $t^1,\ldots,t^p\in\ccmp$, there exists $t\in\ccmp$ such that $q''(t'')+)=(q_k(t^k))_{k\in [p]}=(\ext_{0,p}(q'')+q^e)(t)$ and thus $\gamma_{\ZQ{}{}}(\ext_{0,p}(q'')+q^{e})\supseteq \gamma_{\ZQ{}{}}(q'')\oplus \prod_{k=1}^p \gamma_{\Q^{m+p}}(q_k^{e})$. 
\end{proof}

\begin{proof}[Proof of Theorem~\ref{thjoin}]
\label{proofthjoin}
We only prove that $q\sqsubseteq_{\ZQ{}{}} q \sqcup_{\ZQ{}{}} q'$ i.e. $\gamma_{\ZQ{}{}}(q)\subseteq\gamma_{\ZQ{}{}}(\ext_{0,p}(q'')+q^{e})$. By Lemma~\ref{lemmajoin}, it is the same to prove that 
$\gamma_{\ZQ{}{}}(q)\subseteq \gamma_{\ZQ{}{}}(q'')\oplus \prod_{k=1}^p \gamma_{\Q^{m+p}}(q_k^{e})$. This is equivalent to 
show that for all $t\in\ccm$, there exists $t''\in\ccm$ such that for all $k\in [p]$, $q_k(t)-q''_k(t'')\in \gamma_{\Q}(q_k^e)=\left[\mathbf{b}^{q_k^e},\mathbf{B}^{q_k^e}\right]$. Let for all $k\in [p]$, $r_k=q_k-q_k''\in\Q^m$ and $r'_k=q'_k-q_k''\in\Q^{m}$. From Property~\ref{concabst} and the second point of Property~\ref{extensionprop}, we have 
$\left[\mathbf{b}^{q_k^e},\mathbf{B}^{q_k^e}\right]=\gamma_\Q(r_k)\sint \gamma_\Q(r'_k)$ and from the definition of 
$\sint$, we have $\mathbf{b}^{q_k^e}=\min\left(\mathbf{b}^{r^k},\mathbf{b}^{r'_k}\right)$ and $\mathbf{B}^{q_k^e}=\max\left(\mathbf{B}^{r^k},\mathbf{B}^{r'_k}\right)$. Finally, it suffices to show that 
that for all $t\in\ccm$, there exists $t''\in\ccm$ such that for all $k\in [p]$, $q_k(t)-q''_k(t'')\geq  \min\left(\mathbf{b}^{r^k},\mathbf{b}^{r'_k}\right)$ and $q_k(t)-q''_k(t'')\leq  \max\left(\mathbf{B}^{r^k},\mathbf{B}^{r'_k}\right)$.
Let $t\in \ccm$ and let us take $t''=t$, we have for all $k\in [p]$,  
$\min(\mathbf{b}^{r_k},\mathbf{b}^{r'_k})\leq\mathbf{b}^{r_k} \leq q_k(t'')-q''_k(t'')=q_k(t)-q''_k(t'')\leq \mathbf{B}^{r_k}\leq \max\left(\mathbf{B}^{r^k},\mathbf{B}^{r'_k}\right)$.
\end{proof}

\subsection*{Proofs of bounds improvements}
%Recall that for all $X\in\mathbb{S}_n^+$ of rank one, there exists $x\in\rr^n$ such that $X=xx^\intercal$ and 
%a matrix of the form $X=xx^\intercal$ with $x\in\rr^n$ is semidefinite positive of rank one.   
We remind a classical result on semidefinite positive matrices:
($X\in\mathbb{S}_n^+\ \land\ \operatorname{rk}(X)=1$) if and only if $\exists\, x\in\rr^n\ \st\ X=xx^\intercal$. 
\begin{proof}[Proof of Lemma~\ref{translation}]
\label{translationproof}
Recall that $\ccm=[-1,1]^m\times [-1,1]\times [0,1] \times [-1,0]$.
Let $X\in\mathbb{S}_{m+4}^+$ such that $\operatorname{rk}(X)=1$ and $X\in\CCm$. Since $X\in\mathbb{S}_{m+4}^+$ and $\operatorname{rk}(X)=1$, there exists $u\in\rr^{m+3}$ and $v\in\rr$ such that $X=xx^\intercal$ with $x=(u\ v)$ and thus for all $i,j\in [m+4]$, $X_{i,j}=x_i x_j$. Now since $X\in\CCm$, from Constraint~\eqref{eq_cont:6}, $x_{m+4}x_{m+4}=v^2=1$ and then $v\in\{-1,1\}$. Using the fact that $X=(-x)(-x)^\intercal$, we can choose $v=1$. Now from Constraint~\eqref{eq_cont:3}, for all $i\in[m+3]$, $X_{i,m+4}=x_i x_{m+4}=x_i=u_i\in [-1,1]$. Finally, from Constraint~\eqref{eq_cont:4}, we get $X_{m+2,m+4}=x_{m+2} x_{m+4}=x_{m+2}=u_{m+2}\in [0,1]$ and 
from Constraint~\eqref{eq_cont:5}, $X_{m+3,m+4}=x_{m+3} x_{m+4}=x_{m+3}=u_{m+3}\in [-1,0]$. We conclude that $u\in\cc^{m+3}$.

Now let us take $X$ of the form $X=xx^\intercal$ such that $x=(u\ v)$ with $u\in\cc^{m+3}$ and $x_{m+4}=v=1$. 
Since for all $i,j\in[m+4]$, $X_{i,j}=x_i x_j$ and $u\in\cc^{m+3}$, we have readily Constraints~\eqref{X_constraint}.
\end{proof}

\begin{proof}[Proof of Theorem~\ref{improvboundsth}]
\label{proofboundsimprov}
We only prove the proposition for the upper bounds, the proof for the lower bounds can be done similarly. 
Let $X\succeq 0$ be in $\CCm$. We can write $X$ as a the following block matrix (the notations around the matrix indicates the sizes of the blocks):

\vspace{8pt}

\[
\begin{array}{cl}
\multirow{3}{*}{$\begin{pmatrix}\noalign{\vspace{-20pt}}
m & 3 & 1 \\[-3pt]
\overbrace{} & \overbrace{} & \overbrace{} \\
X^A & X^0 & X^b \\
{X^0}^\intercal & X^{00} & X^c\\
{X^b}^\intercal & {X^c}^\intercal & 1 
\end{pmatrix}$} &\big\} m\\
&\big\} 3\\
&\big\} 1\\
\end{array}
\]
We now rely on the fact that $\tr$ is linear and satisfies for all square matrices $M$, $\tr(M^\intercal)=\tr(M)$ and for all matrices $M,N$ such that 
$MN$ and $NM$ are square matrices, $\tr(MN)=\tr(NM)$. Considering this and the symmetry of $X$, we have after simplifications: 
$\tr\left(\Mat^{r_q} X\right)=\tr\left(A^q X^{A}\right)+{b^q}^\intercal X^{b}+{X^{c}}^\intercal (c_{\pm},c_+,c_-)+c^q$.
Constraints~\eqref{eq_cont:1} and~\eqref{eq_cont:3} yield to $\tr\left(A^q X^{A}\right)\leq \sum_{i=1}^m \sum_{\substack{j=1,\ldots,m\\ j\neq i}} \left|(A^q)_{i,j}\right|+\sum_{i=1}^m [(A^q)_{i,i}]^+$. Constraint~\eqref{eq_cont:2} implies that ${b^q}^\intercal X^{b_q}\leq \sum_{j=1}^{m} \left|{b^q}_{j}\right|$ and Constraints~\eqref{eq_cont:4} and~\eqref{eq_cont:5} imply that ${X^{c}}^\intercal (c_{\pm},c_+,c_-)\leq c_{\pm}+c_+$. By summation, we conclude that $\tr\left(\Mat^{r_q} X\right)\leq \mathbf{B}^q_{MT}$ for all $X\succeq 0$ in $\CCm$ and then $\mathbf{B}^q_{SDP}\leq \mathbf{B}^q_{MT}$.
\end{proof}

\subsection*{Experiments}
\label{sec:exp_appendix}
The following table summarizes the numerical results obtained when comparing
interval arithmetics, zonotopic domains with affine arithmetics and the
quadratic extension. The T1P domain is an extension of the affine zonotopes with
additional linear constraints over error terms. 

Numerical values have been truncated to ease their display. However, even in
presence of similar values, the
highlighted ones were more precise before the truncation.

%\LTXtable{\textwidth}{

%\begin{sideways}
%\begin{minipage}{\textheight}
%\begin{landscape}
\centering
\begin{longtable}
{|l
|>{$}p{2cm}<{$}|>{$}c<{$}%|c
|>{$}p{2.5cm}<{$}|>{$}c<{$}%|c
|>{$}p{2cm}<{$}|>{$}c<{$}%|c
|>{$}p{2cm}<{$}|>{$}c<{$}%|c
%|c|c|c
%|c|c|c
%|c|c|c
|}
%\caption[Feasible triples for a highly variable Grid]{Feasible triples for 
%highly variable Grid, MLMMH.} \label{grid_mlmmh} 
\hline
& 
\multicolumn{2}{c|}{Intervals} & 
\multicolumn{2}{c|}{Affine Z.} &
\multicolumn{2}{c|}{Quad. Z.} %& 
%\multicolumn{2}{c|}{Quad. Z. 2} %& 
%\multicolumn{3}{c|}{Quad. Z. 3} & 
%\multicolumn{3}{c|}{Quad. Z. 4} & 
& \multicolumn{2}{c|}{T1P} 
\\
 \cline{2-9}
& val & ms 
& val  & ms 
& val  & ms 
& val  & ms 
%& val & w. & ms 
%& val & w. & ms 
%& val & w. & ms 
\\
\hline
\endhead
\hline
\endfoot
atan
& [-1.91,\allowbreak 1.91] & 7
& [-1.91,\allowbreak 1.91] & 11
& \best{[-1.00,\allowbreak 1.00]} & 3
& [-1.34,\allowbreak 1.34] & 13
\\
atan10
& [-1.91,\allowbreak 1.91] & 0
& [-2.36,\allowbreak 2.36] & 6
& [-1.59,\allowbreak 1.59] & 8
& \best{[-1.47,\allowbreak 1.47]} & 8
\\
stolfi1
& [0.,\allowbreak \infty] & 0
& [-\infty,\allowbreak +\infty] & 0
& [-0.85,\allowbreak 3.25] & 5
& \best{[0.,\allowbreak 3.60]} & 6
\\
stolfi2
& [0.,\allowbreak 3.60] & 7
& [-\infty,\allowbreak +\infty] & 3
& [-\infty,\allowbreak +\infty] & 6
& \best{[-0.,\allowbreak 2.51]} & 5
\\
stolfi3
& [0.,\allowbreak 5.38] & 6
& [-2.98,\allowbreak 10.81] & 3
& [-0.62,\allowbreak 3.24] & 4
& \best{[-0.,\allowbreak 1.75]} & 6
\\
stolfi4
& [0.,\allowbreak 2.23] & 11
& [-\infty,\allowbreak +\infty] & 8
& [-0.33,\allowbreak 3.26] & 0
& \best{[0.35,\allowbreak 1.71]} & 8
\\
stolfi5
& [0.,\allowbreak 2.18] & 0
& [-0.42,\allowbreak 3.03] & 4
& [0.08,\allowbreak 2.38] & 4
& \best{[0.34,\allowbreak 1.55]} & 7
\\
stolfi6
& [0.,\allowbreak 2.18] & 3
& [-0.71,\allowbreak 2.91] & 9
& [0.11,\allowbreak 2.30] & 11
& \best{[0.38,\allowbreak 1.65]} & 8
\\
stolfi7
& [0.,\allowbreak 1.89] & 6
& [0.19,\allowbreak 2.23] & 3
& [0.29,\allowbreak 1.97] & 6
& \best{[0.45,\allowbreak 1.46]} & 8
\\
stolfi8
& [-0.,\allowbreak 1.82] & 4
& [0.20,\allowbreak 2.] & 0
& [0.29,\allowbreak 1.88] & 3
& \best{[0.47,\allowbreak 1.40]} & 3
\\
stolfi9
& [0.,\allowbreak 1.73] & 10
& [0.26,\allowbreak 1.89] & 6
& [0.33,\allowbreak 1.75] & 10
& \best{[0.47,\allowbreak 1.40]} & 7
\\
stolfi10
& [0.,\allowbreak 1.71] & 11
& [0.32,\allowbreak 1.82] & 0
& [0.36,\allowbreak 1.75] & 5
& \best{[0.49,\allowbreak 1.38]} & 0
\\
stolfi11
& [0.,\allowbreak 1.67] & 3
& [0.34,\allowbreak 1.74] & 3
& [0.39,\allowbreak 1.67] & 4
& \best{[0.51,\allowbreak 1.36]} & 15
\\
stolfi13
& [0.,\allowbreak 1.61] & 10
& [0.38,\allowbreak 1.67] & 0
& [0.41,\allowbreak 1.61] & 8
& \best{[0.51,\allowbreak 1.35]} & 6
\\
stolfi30
& [0.35,\allowbreak 1.43] & 10
& [0.48,\allowbreak 1.44] & 13
& [0.48,\allowbreak 1.43] & 16
& \best{[0.53,\allowbreak 1.31]} & 11
\\
stolfi40
& [0.40,\allowbreak 1.40] & 11
& [0.49,\allowbreak 1.40] & 15
& [0.50,\allowbreak 1.40] & 14
& \best{[0.53,\allowbreak 1.31]} & 15
\\
stolfi50
& [0.43,\allowbreak 1.38] & 12
& [0.50,\allowbreak 1.38] & 24
& [0.50,\allowbreak 1.38] & 26
& \best{[0.53,\allowbreak 1.30]} & 18
\\
stolfi55
& [0.44,\allowbreak 1.37] & 14
& [0.51,\allowbreak 1.37] & 28
& [0.51,\allowbreak 1.37] & 27
& \best{[0.53,\allowbreak 1.30]} & 19
\\
stolfi100
& [0.48,\allowbreak 1.34] & 27
& [0.52,\allowbreak 1.34] & 69
& [0.52,\allowbreak 1.34] & 52
& \best{[0.54,\allowbreak 1.30]} & 30
\\
stolfi200
& [0.51,\allowbreak 1.32] & 40
& [0.53,\allowbreak 1.32] & 272
& [0.53,\allowbreak 1.32] & 209
& \best{[0.54,\allowbreak 1.30]} & 48
\\
stolfi300
& [0.52,\allowbreak 1.31] & 50
& [0.53,\allowbreak 1.31] & 710
& [0.53,\allowbreak 1.31] & 439
& \best{[0.54,\allowbreak 1.30]} & 90
\\
stolfi400
& [0.52,\allowbreak 1.31] & 67
& [0.53,\allowbreak 1.31] & 1359
& [0.53,\allowbreak 1.31] & 842
& \best{[0.54,\allowbreak 1.30]} & 111
\\
Ln1px
& [-0.00,\allowbreak 0.46] & 4
& [-0.00,\allowbreak 0.40] & 8
& \best{[-0.00,\allowbreak 0.40]} & 5
& [-5.14e-05,\allowbreak 0.40] & 3
\\
householder\_orig
& \best{[0.11,\allowbreak 0.11]} & 5
& [0.11,\allowbreak 0.11] & 7
& [0.11,\allowbreak 0.11] & 0
& [0.11,\allowbreak 0.11] & 2
\\
householder\_orig
& [0.17,\allowbreak 0.18] & 0
& [0.17,\allowbreak 0.18] & 8
& [0.17,\allowbreak 0.18] & 0
& \best{[0.17,\allowbreak 0.18]} & 7
\\
householder\_orig
& [0.21,\allowbreak 0.24] & 4
& [0.21,\allowbreak 0.24] & 7
& [0.21,\allowbreak 0.24] & 0
& \best{[0.21,\allowbreak 0.24]} & 9
\\
householder\_orig
& [0.17,\allowbreak 0.29] & 6
& [0.22,\allowbreak 0.25] & 6
& \best{[0.22,\allowbreak 0.24]} & 11
& [0.22,\allowbreak 0.25] & 3
\\
householder\_orig
& [0.03,\allowbreak 0.42] & 0
& [0.22,\allowbreak 0.25] & 10
& \best{[0.22,\allowbreak 0.24]} & 3
& [0.22,\allowbreak 0.25] & 3
\\
householder\_orig
& [-0.90,\allowbreak 1.66] & 7
& [0.22,\allowbreak 0.25] & 12
& \best{[0.22,\allowbreak 0.24]} & 5
& [0.22,\allowbreak 0.25] & 0
\\
householder\_orig
& [-1117.82,\allowbreak 1899.48] & 11
& [0.22,\allowbreak 0.25] & 26
& \best{[0.22,\allowbreak 0.24]} & 7
& [0.22,\allowbreak 0.25] & 3
\\
householder\_orig
& [-2.18e^{+18},\allowbreak 3.70e^{+18}] & 3
& [0.22,\allowbreak 0.25] & 16
& \best{[0.22,\allowbreak 0.24]} & 3
& [0.22,\allowbreak 0.25] & 14
\\
householder\_orig
& [-6.19e^{+94},\allowbreak 1.05e^{+95}] & 4
& [0.22,\allowbreak 0.25] & 32
& \best{[0.22,\allowbreak 0.25]} & 12
& [0.22,\allowbreak 0.25] & 11
\\
householder\_orig
& [-\infty,\allowbreak \infty] & 4
& [0.22,\allowbreak 0.25] & 39
& \best{[0.22,\allowbreak 0.25]} & 7
& [0.22,\allowbreak 0.25] & 4
\\
householder\_orig
& [-\infty,\allowbreak \infty] & 8
& [0.22,\allowbreak 0.25] & 34
& \best{[0.22,\allowbreak 0.25]} & 6
& [0.22,\allowbreak 0.25] & 7
\\
householder\_orig
& [-\infty,\allowbreak \infty] & 7
& [0.22,\allowbreak 0.25] & 37
& \best{[0.22,\allowbreak 0.25]} & 4
& [0.22,\allowbreak 0.25] & 8
\\
householder\_orig
& [-\infty,\allowbreak \infty] & 11
& [0.22,\allowbreak 0.25] & 20
& \best{[0.22,\allowbreak 0.25]} & 7
& [0.22,\allowbreak 0.25] & 3
\\
householder\_orig
& [-\infty,\allowbreak \infty] & 0
& [0.22,\allowbreak 0.25] & 22
& [0.22,\allowbreak 0.25] & 6
& \best{[0.22,\allowbreak 0.25]} & 7
\\
householder\_orig
& [-\infty,\allowbreak \infty] & 3
& [0.22,\allowbreak 0.25] & 11
& [0.22,\allowbreak 0.25] & 7
& \best{[0.22,\allowbreak 0.25]} & 10
\\
householder\_orig
& [-\infty,\allowbreak \infty] & 4
& [0.22,\allowbreak 0.25] & 21
& [0.23,\allowbreak 0.26] & 3
& \best{[0.22,\allowbreak 0.25]} & 9
\\
householder\_orig
& [-\infty,\allowbreak \infty] & 7
& [0.22,\allowbreak 0.25] & 21
& [0.24,\allowbreak 0.27] & 6
& \best{[0.22,\allowbreak 0.25]} & 11
\\
householder
& \best{[0.11,\allowbreak 0.11]} & 7
& [0.11,\allowbreak 0.11] & 2
& [0.11,\allowbreak 0.11] & 4
& [0.11,\allowbreak 0.11] & 7
\\
householder
& [0.17,\allowbreak 0.18] & 5
& [0.17,\allowbreak 0.18] & 10
& \best{[0.17,\allowbreak 0.18]} & 10
& [0.17,\allowbreak 0.18] & 3
\\
householder
& [0.21,\allowbreak 0.24] & 4
& [0.21,\allowbreak 0.24] & 7
& \best{[0.21,\allowbreak 0.24]} & 5
& [0.21,\allowbreak 0.24] & 3
\\
householder
& [0.17,\allowbreak 0.29] & 0
& [0.21,\allowbreak 0.25] & 6
& \best{[0.22,\allowbreak 0.25]} & 11
& [0.21,\allowbreak 0.25] & 5
\\
householder
& [0.03,\allowbreak 0.42] & 4
& [0.20,\allowbreak 0.26] & 12
& \best{[0.21,\allowbreak 0.25]} & 20
& [0.21,\allowbreak 0.26] & 11
\\
householder
& [-0.90,\allowbreak 1.66] & 6
& [0.19,\allowbreak 0.27] & 4
& \best{[0.21,\allowbreak 0.25]} & 16
& [0.20,\allowbreak 0.26] & 12
\\
householder
& [-1117.82,\allowbreak 1899.48] & 0
& [0.17,\allowbreak 0.29] & 11
& \best{[0.21,\allowbreak 0.25]} & 13
& [0.19,\allowbreak 0.27] & 4
\\
householder
& [-2.18e^{+18},\allowbreak 3.70e^{+18}] & 3
& [0.12,\allowbreak 0.34] & 10
& \best{[0.21,\allowbreak 0.25]} & 9
& [0.17,\allowbreak 0.29] & 7
\\
householder
& [-6.19e^{+94},\allowbreak 1.05e^{+95}] & 3
& [-0.16,\allowbreak 0.64] & 6
& \best{[0.21,\allowbreak 0.26]} & 16
& [0.14,\allowbreak 0.33] & 10
\\
householder
& [-\infty,\allowbreak \infty] & 3
& [-14.08,\allowbreak 14.55] & 8
& \best{[0.20,\allowbreak 0.26]} & 15
& [0.00,\allowbreak 0.47] & 9
\\
householder
& [-\infty,\allowbreak \infty] & 0
& [-98003775.36,\allowbreak 98003775.83] & 11
& \best{[0.20,\allowbreak 0.27]} & 22
& [-1.45,\allowbreak 2.24] & 9
\\
householder
& [-\infty,\allowbreak \infty] & 7
& [-1.35e^{+42},\allowbreak 1.35e^{+42}] & 11
& \best{[0.18,\allowbreak 0.30]} & 16
& [-5335.40,\allowbreak 8226.73] & 10
\\
householder
& [-\infty,\allowbreak \infty] & 4
& [-6.88e^{+212},\allowbreak 6.88e^{+212}] & 9
& \best{[0.15,\allowbreak 0.35]} & 11
& [-3.66e^{+21},\allowbreak 5.65e^{+21}] & 5
\\
householder
& [-\infty,\allowbreak \infty] & 3
& [-\infty,\allowbreak +\infty] & 13
& \best{[0.04,\allowbreak 0.52]} & 11
& [-5.61e^{+110},\allowbreak 8.65e^{+110}] & 7
\\
householder
& [-\infty,\allowbreak \infty] & 5
& [-\infty,\allowbreak +\infty] & 7
& \best{[-1.41,\allowbreak 2.27]} & 11
& [-\infty,\allowbreak \infty] & 6
\\
householder
& [-\infty,\allowbreak \infty] & 5
& [-\infty,\allowbreak +\infty] & 10
& \best{[-2570.28,\allowbreak 2582.02]} & 11
& [-\infty,\allowbreak \infty] & 16
\\
householder
& [-\infty,\allowbreak \infty] & 0
& [-\infty,\allowbreak +\infty] & 10
& \best{[-2.85e^{+18},\allowbreak 2.85e^{+18}]} & 7
& [-\infty,\allowbreak \infty] & 6
\\
householder
& [-\infty,\allowbreak \infty] & 5
& [-\infty,\allowbreak +\infty] & 17
& \best{[-4.62e^{+93},\allowbreak 4.62e^{+93}]} & 18
& [-\infty,\allowbreak \infty] & 11
\\
controller
& [-0.21,\allowbreak 0.21] & 7
& \best{[-0.15,\allowbreak 0.15]} & 53
& \best{[-0.15,\allowbreak 0.15]} & 21
& \best{[-0.15,\allowbreak 0.15]} & 10
\\
controller
& [-0.42,\allowbreak 0.42] & 12
& \best{[-0.19,\allowbreak 0.19]} & 191
& \best{[-0.19,\allowbreak 0.19]} & 20
& \best{[-0.19,\allowbreak 0.19]} & 18
\\
controller
& [-0.65,\allowbreak 0.65] & 20
& \best{[-0.20,\allowbreak 0.20]} & 518
& \best{[-0.20,\allowbreak 0.20]} & 49
& \best{[-0.20,\allowbreak 0.20]} & 21
\\
controller
& [-1.25,\allowbreak 1.25] & 30
& \best{[-0.23,\allowbreak 0.23]} & 1865
& \best{[-0.23,\allowbreak 0.23]} & 74
& [-0.23,\allowbreak 0.23] & 35
\\
SinCos
& [0.86,\allowbreak 1.18] & 8
& [0.99,\allowbreak 1.01] & 10
& [0.99,\allowbreak 1.00] & 7
& \best{[0.99,\allowbreak 1.00]} & 2
\\

\end{longtable}
%\end{landscape}
%\end{minipage}
%\end{sideways}
%%% Local Variables: 
%%% mode: latex
%%% TeX-master: "submission"
%%% End: 

\end{document}